\documentclass[11pt]{article}
 
\usepackage[margin=1in]{geometry}
\usepackage[utf8]{inputenc}
\usepackage{microtype}
\usepackage{amsmath,amsthm,amssymb, color, tikz, mathtools}
\usepackage{caption}

\usepackage{algorithm}
\usepackage{algpseudocode}

\newcommand{\Domain}{\cD}
\newcommand{\otilde}{\circ}
\renewcommand{\bar}{\overline}

\usetikzlibrary{positioning}
\usetikzlibrary{shapes}
\usetikzlibrary{decorations.pathreplacing}
\usetikzlibrary{arrows.meta}
\usepackage{enumerate} 
\usepackage[pagebackref]{hyperref} 
\hypersetup{colorlinks,linkcolor={blue},citecolor={blue},urlcolor={blue}}
\usepackage{algorithm}
\usepackage{algpseudocode}

\DeclareMathOperator*{\argmax}{arg\,max}

\newcommand{\II}{\mathbb{I}}

\newcommand{\NN}{\mathbb{N}}

\newcommand{\cA}{\mathcal{A}}

\newcommand{\cD}{\mathcal{D}}

\newcommand{\cR}{\mathcal{R}}
\newcommand{\cS}{\mathcal{S}}
\newcommand{\cT}{\mathcal{T}}

\newcommand{\cX}{\mathcal{X}}
\newcommand{\cY}{\mathcal{Y}}

\newcommand{\sD}{\mathsf{D}}

\newcommand{\sR}{\mathsf{R}}

\newcommand{\set}[1]{\left\{ #1 \right\}}

\newcommand{\ceil}[1]{\left\lceil #1 \right\rceil}
\newcommand{\polylog}{\mathrm{polylog }}

\usepackage{bbm}

\newcommand{\zone}{\set{0,1}}

\newcommand{\OR}{\mathsf{OR}}
\newcommand{\AND}{\mathsf{AND}}

\newcommand{\SUM}{\mathsf{SUM}}
\newcommand{\maxsum}{\mathsf{MAXSUM}}

\newcommand{\king}{\mathsf{KING}}
\newcommand{\IndexKing}{\mathsf{IndexKING}}
\newcommand{\tIndexKing}{\mathsf{t}{\text -}\mathsf{IndexKING}}
\newcommand{\PMF}{\mathsf{PMF}}

\newcommand{\src}{\mathsf{SRC}}
\newcommand{\srcdec}{\mathsf{SRC}^{\textnormal{dec}}}
\newcommand{\cis}{\mathsf{CIS}}

\newcommand{\rank}{\mathrm{rank}}
\newcommand{\dtsize}{\mathsf{DTsize}}
\newcommand{\val}{\mathsf{val}}
\newcommand{\Dcc}{\mathsf{D}^{\mathsf{cc}}}
\newcommand{\Rcc}{\mathsf{R}^{\mathsf{cc}}}
\newcommand{\Qcc}{\mathsf{Q}^{\mathsf{cc}}}

\newcommand{\D}{\mathsf{D}}
\newcommand{\R}{\mathsf{R}}
\newcommand{\Q}{\mathsf{Q}}

\newcommand{\DISJ}{\mathsf{DISJ}}
\newcommand{\GreaterThan}{\mathsf{GT}}

\newtheorem{theorem}{Theorem}[section]
\newtheorem{corollary}[theorem]{Corollary}

\newtheorem{lemma}[theorem]{Lemma}
\newtheorem{claim}[theorem]{Claim}
\newtheorem{defi}[theorem]{Definition}

\newtheorem{observation}[theorem]{Observation}
\newtheorem{proposition}[theorem]{Proposition}

\newcommand{\ket}[1]{|#1\rangle}

\newcommand{\cbra}[1]{\left\{#1\right\}}

\newcommand{\eps}{\varepsilon}
\renewcommand{\epsilon}{\varepsilon}
\newcommand{\MOD}{\mathsf{MOD}}

\allowdisplaybreaks

\title{On the communication complexity of finding a king in a tournament}
\author{Nikhil S.~Mande\thanks{University of Liverpool, UK. {\tt nikhil.mande@liverpool.ac.uk}}
\and 
Manaswi Paraashar\thanks{University of Copenhagen, Denmark. 
{\tt manaswi.isi@gmail.com}
}
\and
Swagato Sanyal\thanks{Indian Institute of Technology, Kharagpur, India. 
{\tt swagato@cse.iitkgp.ac.in}
}
\and
Nitin Saurabh\thanks{Indian Institute of Technology, Hyderabad, India. 
{\tt nitin@cse.iith.ac.in}
}
}
\date{}

\begin{document}

\maketitle

\begin{abstract}
    A tournament is a complete directed graph. A king in a tournament is a vertex $v$ such that every other vertex is reachable from $v$ via a path of length at most 2. It is well known that every tournament has at least one king. In particular, a maximum out-degree vertex is a king.
    The tasks of finding a king and a maximum out-degree vertex in a tournament has been relatively well studied in the context of query complexity. We study the \emph{communication complexity} of finding a king, of finding a maximum out-degree vertex, and of finding a source (if it exists) in a tournament, where the edges are partitioned between two players. The following are our main results for $n$-vertex tournaments:
    \begin{itemize}
        \item The deterministic communication complexity of finding a source (if it exists, or outputting that there is no source) is $\widetilde{\Theta}(\log^2 n)$. 
        \item The deterministic and randomized communication complexities of finding a king are $\Theta(n)$. The quantum communication complexity of finding a king is $\widetilde{\Theta}(\sqrt{n})$.
        \item The deterministic, randomized and quantum communication complexities of finding a maximum out-degree vertex are $\Theta(n \log n), \widetilde{\Theta}(n)$ and $\widetilde{\Theta}(\sqrt{n})$, respectively.
    \end{itemize}
    Our upper bounds above hold for all partitions of edges, and the lower bounds for a specific partition of the edges. To show the first bullet above, we show, perhaps surprisingly, that the communication task of finding a source in a tournament is \emph{equivalent} to the well-studied Clique vs.~Independent Set problem on undirected graphs. Our communication bounds for finding a source then follow from known bounds on the communication complexity of the Clique vs.~Independent Set problem.
    In view of this equivalence, we can view the communication task of finding a king in a tournament to be a natural generalization of the Clique vs.~Independent Set problem.
    
    One of our lower bounds uses a fooling-set based argument, and all our other lower bounds follow from carefully-constructed reductions from Set-Disjointness.
    An interesting point to note here is that while the deterministic query complexity of finding a king has been open for over two decades, we are able to essentially resolve the complexity of this problem in a model (communication complexity) that is usually harder to analyze than query complexity.
    In addition, we give tight bounds on the randomized query complexity of finding a king, exactly determine its decision tree rank, and give near-tight bounds on the decision tree size of finding a king.
\end{abstract}

\section{Introduction}
\label{sec: Introduction}

Graph problems have been very widely studied through the lens of query and communication complexity. In the most natural query setting, an algorithm has query access to an oracle that on being input a pair of vertices, outputs whether or not an edge exists between those vertices. In the basic communication complexity setup for graph problems, two parties, say Alice and Bob, are given the information about the edges in $E_1$ and $E_2$, respectively, where $E_1$ and $E_2$ are disjoint subsets of all possible edges in the underlying graph. Their task, just as in the query model, is to jointly solve a known graph problem on the graph formed by the edges in $E_1 \cup E_2$. Several interesting results are known in these basic query and communication settings in the deterministic, randomized and quantum models, see, for example,~\cite{BFS86, HMT88, DP89, IKLSW12, Nisan21, BN21, BBEMN22} and the references therein.

A prime example of a graph problem whose query complexity and communication complexities have been widely studied is \emph{Graph Connectivity}. The randomized and quantum communication complexities of this problem are known to be $O(n\log n)$ and $\Omega(n)$. This gap has been open for a long time, and the question of closing it has been explicitly asked~\cite{IKLSW12, HMT88}. On the other hand, its deterministic communication complexity is known to be $\Theta(n\log n)$~\cite{HMT88}.

A graph problem that has been extensively studied in the context of communication complexity is the Clique vs.~Independent Set (CIS) problem~\cite{Yan91, Goos15, GPW15, BBGJK21}. The CIS problem is parametrized by a graph $G = ([n], E)$, known to both Alice and Bob. Alice is given $C \subseteq [n]$ that forms a clique in $G$, Bob is given $I \subseteq [n]$ that forms an independent set in $G$, and their task is to determine whether or not $C \cap I = \emptyset$. Note that if $C \cap I \neq \emptyset$, then it must be the case that $|C \cap I| = 1$.
It was long known that the communication complexity of CIS is $O(\log^2 n)$ for all graphs $G$. More than two decades after this upper bound was discovered, a near-matching lower bound of $\widetilde{\Omega}(\log^2 n)$ was shown to hold for a particular $G$, in a culmination of a long line of work~\cite{KLO99, HS12, Ama14, SA14, Goos15, GPW15}.

\begin{theorem}[{\cite{Yan91}, \cite[Theorem~1.2]{GPW15}}]\label{thm: cis known}
    Let $G$ be an $n$-vertex graph. Then, $\Dcc(\cis_G) = O(\log^2 n)$. Furthermore, there exists an $n$-vertex graph $G$ such that $\Dcc(\cis_G) = \widetilde{\Omega}(\log^2 n)$.
\end{theorem}
This lower bound on the communication complexity of CIS also gives the currently-best-known lower bounds for the famous log-rank conjecture~\cite{LS88}.
We remark that the upper bound above also holds if the task is to output the label of the unique intersection of $C$ and $I$ if $C \cap I \neq \emptyset$.

Switching gears slightly, we now discuss communication complexity on complete directed graphs. A tournament on $n$ vertices is a complete directed graph on $n$ vertices. Throughout this paper, we will view a $n$-vertex tournament as a string $G \in \zone^{\binom{n}{2}}$, where the indices are labeled by pairs $\cbra{i < j \in [n]}$ and $G_{i, j} = 1$ means the edge between vertices $i$ and $j$ is directed from $i$ to $j$.
In the most natural communication complexity setting here, Alice owns a subset $E$ of the edges (i.e., she knows these edge directions), Bob owns the remaining edges, and their goal is to jointly solve a known task on the underlying tournament.
We study the communication complexity of finding a source in a tournament if it exists. That is, Alice and Bob should either output that no source exists, or output the label of the (unique) source.
Denote this task as $\src_E$. Surprisingly, we show that this task is \emph{equivalent} to the CIS problem on undirected graphs.

\begin{theorem}\label{thm: cis equivalent to src}~
    \begin{itemize}
        \item For all $n$-vertex graphs $G = ([n], E)$, $\Dcc(\cis_G) \leq \Dcc(\src_{E}) + O(\log n)$.
        \item For all subsets of edges $E$ of the complete $n$-vertex graph, there exists an $n$-vertex graph $G$ such that $\Dcc(\src_{E}) \leq \Dcc(\cis_G)$.
    \end{itemize}
\end{theorem}

Using known near-tight bounds on the communication complexity of CIS (Theorem~\ref{thm: cis known}), Theorem~\ref{thm: cis equivalent to src} immediately yields the following corollary which gives near-tight bounds on the communication complexity of finding a source in a tournament.
\begin{corollary}\label{cor: CC of source}
        For all subsets $E$ of the edges of a complete $n$-vertex graph, the deterministic communication complexity of finding a source of a tournament if it exists, or outputting that there is no source (where Alice knows the edge directions of edges in $E$ and Bob knows the edge directions of the remaining edges) is
    \begin{align*}
        \Dcc(\src_E) & = O(\log^2 n).
    \end{align*}
    Furthermore, there exists a subset $E$ of edges of the complete $n$-vertex graph such that the deterministic communication complexity of finding a source is
    \begin{align*}
        \Dcc(\src_E) & = \widetilde{\Omega}(\log^2 n).
    \end{align*}
\end{corollary}

Motivated to find a ``most-dominant vertex'' in a tournament, Landau defined the notion of a \emph{king} in a tournament~\cite{landau1953dominance}. A king in a tournament is a vertex $v$ such that every other vertex $w$ is either reachable via a path of length 1 or length 2 from $v$. While it is easy to see that there are tournaments that do not have a source, it is also easy to show that every tournament has a king~\cite{landau1953dominance, maurer1980king}. If a tournament has a source, then it is a unique king in the tournament. In view of this, a natural generalization of $\src_E$ (and hence CIS, in view of Theorem~\ref{thm: cis equivalent to src}) is the communication task of finding a king in a tournament.

We remark here that the deterministic query complexity of finding a king in an $n$-vertex tournament is still unknown, and the state-of-the-art bounds are $\Omega(n^{4/3})$ and $O(n^{3/2})$, and are from over 2 decades ago~\cite{SSW03}.
Recently,~\cite{MPS23} essentially resolved the randomized and quantum query complexities of this problem: they showed that the randomized query complexity of finding a king in an $n$-vertex tournament is $\widetilde{\Theta}(n)$, and the quantum query complexity is $\widetilde{\Theta}(\sqrt{n})$. The complexity of finding a king and natural variants of it have also been fairly well-studied in different contexts~\cite{SSW03, AFHN16, BJRS22, LRT22}.

We consider the communication complexity of finding a king in an $n$-vertex tournament (the edge partition will be clear from context), denoting this task by $\king_n$. Perhaps surprisingly, while resolving the query complexity of finding a king in a tournament seems hard, we are able to essentially resolve its asymptotic deterministic, randomized and quantum communication complexities.
\begin{theorem}
\label{thm: CC of Kings}
    For all disjoint partitions $E_1, E_2$ of the edges of a tournament, the deterministic, randomized and quantum communication complexities of finding a king (where Alice knows the edge directions of edges in $E_1$ and Bob knows the edge directions of edges in $E_2$) are as follows:
    \begin{align*}
        \Dcc(\king_n) = O(n),\quad
        \Rcc(\king_n) = O(n),\quad
        \Qcc(\king_n) = \widetilde{O}(\sqrt{n}).
    \end{align*}
    Furthermore, there exists disjoint partition $E_1, E_2$ such that the deterministic, randomized and quantum communication complexities of finding a king are as follows:
    \begin{align*}
        \Dcc(\king_n) = \Omega(n),\quad
        \Rcc(\king_n) = \Omega(n),\quad
        \Qcc(\king_n) = \Omega(\sqrt{n}).
    \end{align*}
\end{theorem}
In order to show our deterministic and randomized upper bounds, we give a cost $O(n)$ deterministic protocol. Our quantum upper bound follows from the quantum query upper bound of~\cite{MPS23} along with a well-known simulation of a quantum query algorithm using a quantum communication protocol~\cite{BCW98}. Our lower bounds follow from a carefully constructed reduction from Set-Disjointness. We sketch our proofs in Section~\ref{sec: sketch of proofs}.

Interestingly, our lower bounds actually hold for tournaments that are promised to have exactly 3 kings. It is well known that a tournament cannot have exactly 2 kings~\cite{maurer1980king}. Thus, the only ``easier'' case than this promised one is that where the input tournament is promised to have exactly one king. This case is handled in Corollary~\ref{cor: CC of source} (it is easy to see that a tournament has a unique king iff the unique king is a source in the tournament).

It is folklore~\cite{landau1953dominance} that a vertex with maximum out-degree in a tournament is also a king in the tournament. Thus, another natural question that arises is: what is the complexity of finding a maximum out-degree vertex? The deterministic and randomized query complexity of this task is known to be $\Theta(n^2)$, and its quantum query complexity is between $\Omega(n)$ and $O(n^{3/2})$~\cite{balasubramanian1997finding, MPS23}. Let $\MOD_n$ denote the search problem of finding a maximum out-degree vertex in an $n$-vertex tournament. We study the communication complexity of $\MOD_n$, again in the natural setting where the edges of the tournament are partitioned between Alice and Bob. We show the following:

\begin{theorem}
\label{thm: CC of MOD}
    For all disjoint partitions $E_1, E_2$ of the edges of a tournament, the deterministic, randomized and quantum communication complexities of finding a king (where Alice knows the edge directions of edges in $E_1$ and Bob knows the edge directions of edges in $E_2$) are as follows:
    \begin{align*}
        \Dcc(\MOD_n) = O(n \log n), \quad \Rcc(\MOD_n) = O(n \log\log n), \quad \Qcc(\MOD_n) = {O}(\sqrt{n} \log n).
    \end{align*}
    Furthermore, there exist disjoint partitions such that the deterministic, randomized and quantum communication complexities of finding a king are as follows:\footnote{The edge partition we use to prove our deterministic lower bound is different from the partition we use to prove our randomized and quantum lower bounds.}
    \begin{align*}
        \Dcc(\MOD_n) = \Omega(n \log n), \quad \Rcc(\MOD_n) = \Omega(n), \quad \Qcc(\MOD_n) = \Omega(\sqrt{n}).
    \end{align*}
\end{theorem}

We direct the reader's attention to the similarity between our communication complexity bounds for $\MOD_n$ and known bounds for the communication complexity of Graph Connectivity mentioned earlier in this section: just like in that case we are able to give tight bounds on the deterministic communication complexity, but our bounds are loose by logarithmic factors in the randomized and quantum settings. Our randomized and quantum lower bounds follow using exactly the same reduction from Set-Disjointness as in Theorem~\ref{thm: CC of Kings}. Our deterministic lower bound follows by a carefully constructed \emph{fooling set} lower bound. We give a sketch of our proofs in the next section.

\subsection{Sketch of proofs}\label{sec: sketch of proofs}
\subsubsection{Equivalence of source-finding and CIS}
We first sketch the proof of Theorem~\ref{thm: cis equivalent to src}, which is the equivalence of finding a source in a tournament and the Clique vs.~Independent Set problem.
Consider a graph $G = ([n], E)$, and an input $C, I$ to the Clique vs.~Independent Set problem. Here Bob is given $C \subseteq [n]$ which is a clique in $G$, and Alice is given $I \subseteq [n]$ which is an independent set in $G$ (we switch the order of inputs for convenience). Alice and Bob construct the following instance to the source-finding problem:
\begin{itemize}
    \item Alice has the edge directions of all edges in $E$, and Bob has the remaining edge directions
    in $\bar{E}$.
    \item Alice constructs her edge directions such that all vertices in $I$ have in-degree $0$ with respect to her edge directions in $E$. This is easy to do since there are no edges between any pair of vertices in $I$. She also ensures that all vertices in $[n] \setminus I$ have in-degree at least $1$, with respect to her edge directions in $E$. 
    She can ensure this if $G$ is a connected graph. (see Section~\ref{sec: Communication complexity of finding a source}.)
    
    \item Just as the above, Bob ensures that all vertices in $C$ have in-degree 0 w.r.t.~$\bar{E}$, and all vertices in $I \setminus C$ have in-degree at least 1 w.r.t.~$\bar{E}$.
\end{itemize}
Using the properties above, it is not hard to show that $s = C \cap I$ iff $s$ is a source in the tournament jointly constructed by Alice and Bob above. This concludes the reduction from CIS to source-finding. 

In the other direction, if Alice is given edge directions for the subset $E$ of edges of the complete $n$-vertex graph, then the underlying graph $G$ that Alice and Bob construct for the CIS problem is $G = ([n], E)$. For the purpose of this reduction, we assume that Alice has an independent set as input to CIS, and Bob has a clique. Alice considers her input independent set $I$ to the CIS problem to be the set of all vertices with in-degree 0 w.r.t.~$E$ (note that these vertices must form an independent set in $G$), and Bob constructs his input clique $C$ to be all vertices with in-degree 0 w.r.t.~his edges (these form a clique w.r.t.~$E$, and hence in $G$). Note that a source in the initial tournament, if it exists, must be a vertex in $I \cap C$ since it must have in-degree 0 both w.r.t.~Alice's and w.r.t.~Bob's edges. Moreover this is the only way in which $I$ intersection $C$ is non-empty.
In other words, $I \cap C \neq \emptyset$ iff there is a source in the initial tournament. This concludes the reduction from source-finding to CIS, and hence Theorem~\ref{thm: cis equivalent to src}.
Known upper bounds and lower bounds on the communication complexity of the Clique vs.~Independent Set problem (Theorem~\ref{thm: cis known}) then yield Corollary~\ref{cor: CC of source}.

Some of our proofs of the lower bounds in Theorems~\ref{thm: CC of Kings} and~\ref{thm: CC of MOD} follow the same outline. In the next section, we sketch our upper bounds, and we sketch our lower bounds in the following section.

\subsubsection{Upper bounds}
We start with ideas behind the upper bounds in Theorem~\ref{thm: CC of Kings}. Recall that the goal is to construct a communication protocol for finding a king a tournament $G \in \zone^{\binom{n}{2}}$ whose edges are partitioned into $E_1$ (with Alice) and $E_2$ (with Bob).

Consider the deterministic communication model. 
In the beginning of each round assume without loss of generality that Alice has a larger number of edges. Alice sends Bob the label of a vertex $v$ with maximum number of out-neighbours in $E_1$ along with the in-neighbourhood of $v$ in $E_1$ as a bit-string. Upon receiving $v$, Bob also sends the in-neighbourhood of $v$ in $E_2$ as a bit-string. Thus both players know the entire in-neighbourhood of $v$ in the entire tournament by the end of the round. The communication cost so far is at most $2n + \log n = O(n)$, where $n$ is the number of vertices in the current tournament.
The players now reduce to finding a king in the in-neighbourhood of $v$, since by~\cite{maurer1980king} (also see Lemma~\ref{lem: maurer king in inneighbour}), this would give a king in the tournament $G$. Since $|E_1| \geq |E_2|$, the number of out-neighbours of $v$ is at least $(n - 1)/4$. This yields a communication protocol of cost $T(n)$ that is described by a recurrence of the form $T(n) \leq T(3n/4) + O(n)$, which is easily seen to give a solution of $T(n) = O(n)$.

The quantum communication protocol for finding a king in $G$ is obtained by simulating the quantum query algorithm due to~\cite{MPS23} (also see Theorem~\ref{thm: MPS23 quantum query algo kings}).~\cite{MPS23} gave an $O(\sqrt{n}~\polylog(n))$ query algorithm, which can be used to obtain a communication protocol with $O(\log n)$-overhead by using the simulation theorem of~\cite{BCW98} (also see Theorem~\ref{thm: bcw}).

We now describe the upper bounds in Theorem~\ref{thm: CC of MOD}. For any tournament $G \in \zone^{\binom{n}{2}}$ and any partition $E_1$, $E_2$ of edges of $G$ given to Alice and Bob, respectively, our goal is to come up with a communication protocol to find a vertex with maximum out-degree. Our upper bounds follow from communication protocols for the following problem: Alice and Bob are given $A \in [n]^n$ and $B \in [n]^n$, respectively.
Their goal is to output an index $i \in [n]$ that maximizes $a_i + b_i$. We call this communication problem $\maxsum_{n,n}$. 
The reduction from $\MOD_n$ to $\maxsum_{n,n}$ is easy to see: Alice and Bob construct $A, B$ to be the vector of in-degrees of all vertices w.r.t.~their edges. Thus a deterministic communication protocol of cost $O(n \log n)$ immediately follows for $\maxsum_{n,n}$: Alice can sends $A$ to Bob, who then computes an answer. We now sketch the randomized upper bound. Let $S = (s_1, \dots, s_n)$ where $s_i = a_i + b_i$. The first observation is that deciding $s_i \geq s_j$ is equivalent to deciding $a_i - a_j \geq b_j - b_i$. The latter can done with cost $O(\log \log n)$ and error at most $1/3$ by using the communication protocol of Greater-Than due to~\cite{Nisan93} (see Theorem~\ref{thm: nisan GT}). 
Thus Alice and Bob have access to a ``noisy'' oracle that decides whether $s_i \geq s_j$, for all $i,j \in [n]$, independently with probability at least $2/3$. Finding $\argmax_{i \in [n]} s_i$ with error probability $1/3$ can be done by making $O(n)$ such queries (due to~\cite{FPRU90}, see Theorem~\ref{thm: max noisy oracle}). This gives a protocol with an overall communication cost of $O(n \log\log n)$. The quantum communication protocol is an application of a result of~\cite{BCW98}, along with a quantum query upper bound for computing argmax (see Theorem~\ref{thm: argmax grover}), see Section~\ref{sec: CC of MOD} for details.

\subsubsection{Lower bounds}

Our intuition for the lower bounds is that a ``hard'' partition of edges between Alice and Bob should be such that every vertex has an equal number of incident edges with Alice and with Bob.
One such natural partition of the edges is as follows: Alice receives the complete tournament restricted to the first $n/2$ vertices and the complete tournament restricted to the last $n/2$ vertices, and Bob receives all of the edges between these vertices. While we are unable to use this partition of edges to prove a lower bound for $\king_n$, we do use it to show a deterministic lower bound for $\MOD_n$. 
Our approach to showing a deterministic communication lower bound for $\MOD_n$ is to construct a large \emph{fooling set} (see Lemma~\ref{lemma: fooling set}). More precisely, for a permutation $\sigma \in S$, where $S$ is a suitably chosen large (size $2^{\Omega(n \log n)}$) subset of $\cS_n$, we construct inputs $A_\sigma, B_\sigma$ to Alice and Bob such that vertex 1 is a unique maximum out-degree vertex for all $\sigma \in S$. 
We also ensure that ``cross-inputs'' $(A_\sigma, B_{\sigma'})$ with $\sigma \neq \sigma'$ lead to vertex 1 not being a maximum out-degree vertex as long as $\sigma$ and $\sigma'$ are far away in the $\ell_\infty$ norm, which we force to be true for all permutations in $S$ by our construction. We refer the reader to Section~\ref{sec: CC of MOD} for technical details.

While we are unable to make the same reduction work to show the communication lower bounds for $\king_n$ (and for good reason, since this argument gives an $\Omega(n \log n)$ lower bound, and there is an $O(n)$ upper bound for the communication complexity of $\king_n$) and randomized and quantum communication lower bounds for $\MOD_n$, 
our partition constructed there has a similar flavor to that above.
A key intermediate function that we consider for showing our remaining lower bounds is a variant of $\king$ inspired by the well-studied Indexing function. Aptly, we name our variant $\IndexKing$, defined below. For a tournament $G \in \zone^{\binom{n}{2}}$ with vertex set $[n]$, and a set $S \subseteq [n]$, we use the notation $G|_S$ to denote the subtournament of $G$ induced on the vertices in $S$.

\begin{defi}
Let $n > 0$ be a positive integer. Define the \emph{$\IndexKing_n$} communication problem as follows:
Alice is given a set $S \subseteq [n]$ and Bob is given a tournament $G \in \zone^{\binom{n}{2}}$ on $n$ vertices. Their goal is to output a king in $G|_S$.
\end{defi}

We consider the restriction of $\IndexKing$ to those inputs where Bob's tournament is a transitive tournament (see Definition~\ref{defi: Transitive Tournament}). We denote this variant by $\tIndexKing$.
A moment's observation (see Observation~\ref{obs: equiv PMF IndexKing}) reveals that this problem is equivalently formulated as follows. We name this version the \emph{Permutation Maximum Finding} problem, defined below, and we feel that this problem is of independent interest.

\begin{defi}[Permutation Maximum Finding]
\label{def: Permutation Maximum Finding}
    Let $n > 0$ be a positive integer. In the \emph{Permutation Maximum Finding} problem, $\PMF_n$, Alice is given as input a subset $S$ of $[n]$, Bob is given a permutation $\sigma \in \cS_n$, and their goal is to output 
    \begin{align*}
        \PMF_n(S, \sigma) = \begin{cases}
            \bot & S = \emptyset\\
            \argmax_{j \in S}\sigma(j) & S \neq \emptyset.
        \end{cases}
    \end{align*}
\end{defi}
Unless explicitly mentioned otherwise, we assume that Alice's input $S$ to $\PMF_n$ is a always non-empty set. In other words, in the $\PMF$ problem, Alice is given a subset of $[n]$, Bob is given a ranking of all elements in $[n]$ (here, $\sigma(i)$ denotes the rank of $i$), and their goal is to find the element in Alice's set that has the largest rank.
\begin{observation}
\label{obs: equiv PMF IndexKing}
Let $n > 0$ be a positive integer. Then,
\begin{align*}
    \mathsf{cost}(\PMF_n) = \mathsf{cost}(\tIndexKing_n),
\end{align*}
where $\mathsf{cost} \in \cbra{\Dcc, \Rcc, \Qcc}$.\footnote{We actually prove the stronger statement that the problems $\PMF_n$ and $\tIndexKing_n$ are equivalent, in the sense that  Alice and Bob need not communicate to go one from one problem to another.}
\end{observation}
For completeness we provide a proof in Section~\ref{sec: Missing Proofs}.

We show that Set-Disjointness reduces to $\PMF$ (see Lemma~\ref{lemma: pmf lower bound} and its proof). The lower bound results for $\PMF$ follow from known results for communication complexity of Set-Disjointness (see Theorem~\ref{thm: CC of Set-Disjointness}).

Next we reduce from $\PMF_n$ to $\king$. Our reduction ensures that an instance $(S, \sigma)$ to $\PMF_n$ gives us a tournament $G_{S, \sigma}$ with the following properties:
\begin{itemize}
    \item The tournament has $3n$ vertices, partitioned into $V_1, V_2, V_3$, of $n$ vertices each, each labeled by elements of $[n]$. The internal edges (edges in $\binom{V_1}{2}, \binom{V_2}{2}$ and $\binom{V_3}{2}$) in each of the partitions are with Bob, and these correspond to transitive tournaments defined by $\sigma$.
    \item The remaining ``cross'' edges are all with Alice, and the directions of these are determined by $S$ (see Figure~\ref{fig: G S sigma} for details).
    \item The tournament $G_{S, \sigma}$ has exactly three kings (which are also the three unique maximum out-degree vertices), one in each $V_i$, and each of these are labeled by $\PMF_n(S, \sigma)$. 
\end{itemize}
Thus finding a king or a maximum out-degree vertex in $G_{S, \sigma}$ amounts to Alice and Bob solving $\PMF_n$, which we've already sketched to be hard via a reduction from Set-Disjointness.
An interesting point to note is that this actually shows a lower bound on the communication complexity of finding a king, even when the input tournament is promised to have exactly three kings. Recall that we showed that finding a king can be done with $O(\log^2 n)$ deterministic communication when an input is promised to have exactly one king (Corollary~\ref{cor: CC of source}). Also it is easy to show using Lemma~\ref{lem: maurer king in inneighbour} that there are no tournaments with exactly two kings. Thus, the ``easiest'' non-trivial case of a promised tournament with exactly three kings is already hard for communication.

\subsection{Other results}
We next turn our attention to the decision tree \emph{size} complexity of $\king$. While most of the relevant literature of finding kings in tournaments deals with minimizing the number of queries to find a king (which is equivalent to minimizing the depth of a decision tree that solves $\king$), none deal with minimizing the \emph{size complexity} of a decision tree that solves $\king$. Logarithm of decision tree size complexity is characterized, upto a log factor in the input size, by the \emph{rank} of the underlying relation, and these are measures that have gained a significant interest in the past few years in various contexts (see, for instance,~\cite{CDMRS23, DM21, CMP22} and the references therein).

While the decision tree depth complexity of $\king_n$ lies between $\Omega(n^{4/3})$ and $O(n^{3/2})$, we show a tight bound of $n - 1$ on $\rank(\king_n)$.

\begin{theorem}\label{thm: rank of king}
    Let $n > 0$ be a positive integer. Then $\rank(\king_n) = n - 1$.
\end{theorem}
As a corollary, Proposition~\ref{prop: rank size} implies a near-tight bound on $\dtsize(\king_n)$.
\begin{corollary}
    Let $n > 0$ be a positive integer. Then,
    \begin{align*}
        \log\dtsize(\king_n) = O(n \log n), \qquad \log\dtsize(\king_n) = \Omega(n).
    \end{align*}
\end{corollary}

It is known (see, for example,~\cite[Lemma A.3]{CDMRS23}) that a lower bound on the communication complexity of a relation under an arbitrary partition of the inputs yields a lower bound on the logarithm of its decision tree size. Thus, a natural attempt to remove the log factors in the above corollary would have been to show a communication lower bound for $\king_n$ of $\Omega(n \log n)$ under some partition of the inputs. However the deterministic communication upper bound in Theorem~\ref{thm: CC of Kings} rules this out.

Finally, we give an asymptotically tight randomized query complexity upper bound for $\king_n$. We  remove the log factors from the previous upper bound~\cite{MPS23} to show an optimal $O(n)$-cost algorithm. Our algorithm is nearly the same as the earlier one, and the upper bound follows just from a more careful analysis.
\begin{theorem}\label{thm: king randomized no log factors}
    Let $n > 0$ be a positive integer. Then $\sR(\king_n) = O(n)$.
\end{theorem}

\section{Preliminaries}
\label{sec: main body prelims}
Let $[n] = \{1, \dots, n\}$. We use the notation $\polylog(n)$ to denote $O(\log(n)^c)$ for some fixed constant $c$. 
A tournament $G \in \zone^{\binom{n}{2}}$ is a complete directed graph on $n$-vertices.
For $v,w \in [n]$ such that $v < w$, if $G_{v, w} = 1$ then there is an out-edge from $v$ to $w$, i.e. $v \rightarrow w$ (otherwise there is an out-edge from $w$ to $v$). In this case we say that $v$ \textit{$1$-step dominates} $w$. Similarly, for $u,w \in [n]$, if there exists a $v \in [n]$ such that $u \rightarrow v$ and $v \rightarrow w$ then we say that $u$ \textit{$2$-step dominates} $w$. Let $S \subseteq [n]$ be such that $v$ $2$-step ($1$-step) dominates $w$ for all $w \in S$. We then say that $v$ $2$-step ($1$-step) dominates $S$. It is easy to see that there are tournaments where no vertex $1$-step dominates all other vertices (such a vertex is called the \textit{source} of $G$). However, it is now folklore that every tournament has a vertex $v$ such that every vertex $w \neq v$ is either $1$-step or $2$-step dominated by $v$. 
Such a vertex is called a \textit{king} of the tournament (see~\cite{landau1953dominance}).

\begin{lemma}[Folklore]\label{lemma: every tournament has a king}
    Let $G \in \zone^{\binom{n}{2}}$ be a tournament. Then there exists a vertex $v \in [n]$ such that $v$ is a king of $G$.
\end{lemma}

For a vertex $v \in [n]$, 
let $N^-(v) = \{w \in [n]: w \rightarrow v \}$ and $N^+(v) = \{w \in [n] : v \rightarrow w \}$. Thus $N^-(v)$ and $N^+(v)$ denote the in-neighbourhood and out-neighbourhood of $v$ in $G$, respectively. The in-degree of $v$, denoted by $d^-(v)$ is defined as $|N^-(v)|$, and similarly the out-degree of $v$ is denoted by $d^+(v)$ and is defined as $|N^+(v)|$. If a vertex has maximum out-degree in the tournament, then that vertex is a king of the tournament (a proof can be found in~\cite{maurer1980king}).

\begin{lemma}[\cite{landau1953dominance}]
\label{lemma: mod vertex king}
    Let $G \in \zone^{\binom{n}{2}}$ be a tournament and $v \in [n]$ be a vertex of maximum out-degree in $G$. Then $v$ is a king in $G$.
\end{lemma}

For $S \subseteq [n]$ let $G|_S$ be the tournament \textit{induced on $S$ by $G$}, i.e. $G|_S$ is a tournament with vertex set as $S$ and direction of edges in $S$ are same as that in $G$.

The following is an important lemma that we use often.

\begin{lemma}[\cite{maurer1980king}]
\label{lem: maurer king in inneighbour}
Let $G \in \zone^{\binom{n}{2}}$ be a tournament and $v \in [n]$. If a vertex $u$ is a king $G|_{N^-(v)}$, then $u$ is a king in $G$.
\end{lemma}

A special class of tournaments is the class of transitive tournaments, which we define next.

\begin{defi}[Transitive Tournament]
\label{defi: Transitive Tournament}
    A tournament $G \in \zone^{\binom{n}{2}}$ is transitive if it satisfies the following property: for all $u,v,w \in [n]$, $u \rightarrow v$ and $v \rightarrow w$ implies $u \rightarrow w$.
\end{defi}
In other words, a transitive tournament is a tournament which is a directed acyclic graph.

\begin{lemma}[Properties of Transitive Tournaments]
\label{lemma: Properties of Transitive Tournaments}
    Let $G \in \zone^{\binom{n}{2}}$ be a transitive tournament. There is an ordering $v_1, \dots, v_n$ of $[n]$ such that 
            \begin{itemize}
                \item $v_1$ is a source vertex and hence a unique king in $G$, and
                \item for all $i \in \{2, \dots, n\}$, $v_i$ is  source vertex in $G|_{[n] \setminus \bigcup_{j=1}^{i-1} \cbra{v_j}}$.
            \end{itemize} 
\end{lemma}
\begin{proof}
    Since $G$ is a directed acyclic graph, a topological sort on the vertices gives a source of the graph. Let this vertex be $v_1$. The vertex $v_i$ is obtained by applying the same argument over the transitive tournament $G|_{[n] \setminus \bigcup_{j=1}^{i-1} \cbra{v_j}}$.
\end{proof}

\subsection{Query and Communication Complexity}
\label{sec: Prelims Query and Communication Complexity}
Let $f \subseteq \Domain^n \times \cR$ be a relation, where $\cD = [k]$ for some finite $k$.
A deterministic query algorithm $\cA$ is an algorithm that knows $f$, is given query access to an unknown $x \in \Domain^n$ (i.e.,~upon ``querying'' $i$, $\cA$ receives $x_i \in \Domain$) and outputs an $r = \cA(x) \in \cR$ such that $(x,r) \in f$ for all $x \in \Domain^n$. The cost of $\cA$ is the number of queries it makes in the worst case over all $x \in \Domain^n$. The deterministic query complexity of $f$, denoted by $\D(f)$, is defined as follows:
\begin{align*}
    \D(f) = \min_{\cA : \cA \textnormal{ computes } f} \textnormal{cost}(\cA).
\end{align*}
A randomized query algorithm $\cA$ is defined similarly to the deterministic query algorithm with a few differences. $\cA$ is given access to random coins, and we say that $\cA$ computes $f$ with error $\epsilon$ if for all $x \in \Domain^n$, $\Pr[(x,\cA(x)) \notin f] \leq \eps$, where the probability is over random coins of $\cA$. Also, the cost of $\cA$ is the number of queries it makes in the worst case over all $x \in \Domain^n$ and the coin tosses. The $\eps$-error randomized query complexity of $f$, denoted by $\R_{\eps}(f)$, is defined as follows:
\begin{align*}
    \R_{\eps}(f) = \min_{\cA : \cA \textnormal{ computes } f \textnormal{ with error } \epsilon} \textnormal{cost}(\cA).
\end{align*}
When $\epsilon = 1/3$, we use the notation $\R(f)$.

We say that an quantum algorithm $\cA$ has quantum query access to $x$ if it has access to the following unitary
\begin{align*}
    Q_x \ket{i}\ket{b} = \ket{i}\ket{(b + x_i)~\mathrm{mod}~k},
\end{align*}
for all $i \in \zone^{\ceil{\log n}}$ and all $b \in [k]$. Note that the second register is a $\ceil{\log k}$ qubit register. A quantum query algorithm $\cA$ that knows $f$ and is given quantum query access to $x$ is said to compute $f$ with error $\eps$ if $\Pr[(x, \cA(x)) \notin f] \leq \eps$ for all $x \in \Domain^n$. The cost of $\cA$ is the number of quantum queries made by it.

\begin{defi}[$\mathsf{ARGMAX}_{k, n}$]
\label{defi: argmax search quantum algo}
    Let $k$ be a positive integer and let $a \in ([k])^n$. Given query access to $a$, find $i \in [n]$ such that $a_i \geq a_j$ for all $j \neq i \in [n]$.
\end{defi}

\begin{theorem}[\cite{durr1996quantum}]
\label{thm: argmax grover}
    There exists a quantum query algorithm for $\mathsf{ARGMAX}_{k, n}$ with query cost $O(\sqrt{n})$.
\end{theorem}

Now we describe the models of communication complexity introduced by Yao~\cite{Yao79, Yao93}. We will restrict to special type of communication problems obtained by composing a relation with a function. Let $\cD_f, \cD_g$  be finite sets, let $f \subseteq \Domain_f^n \times \cR$ be a relation and $g: \cD_g \times \cD_g \to \cD_f$ be a function. Then $f \circ g \subseteq \cbra{\cD_g \times \cD_g}^n \times \cR$ is defined as 
\begin{align} \label{eq: communication goal}
    (x, y, r) \in f \circ g \iff ((g(x^{(1)},y^{(1)}), \dots, g(x^{(n)}, y^{(n)}), r) \in f,
\end{align}
where $x^{(i)}, y^{(i)} \in \cD_g$ for all $i \in [n]$.
In the communication problem corresponding to $f \circ g$, there are two communicating parties, Alice and Bob, who know $f$ and $g$ in advance. Alice is given $x \in (\zone^m)^n$ and Bob is given $y \in (\zone^m)^n$. Their goal is to output $r \in \cR$ such that $(x, y, r) \in f \circ g$. They do this by sending messages (bits in classical case or qubits in quantum case) using a pre-decided protocol $\Pi$. The protocol $\Pi$ can either be deterministic, randomized or quantum depending on the model in consideration.
    \begin{itemize}
        \item In the model of deterministic communication, Alice and Bob want to output a valid $r \in \cR$ for all $x, y$. In this case we say that $\Pi$ computes $(f \otilde g)$. The cost of $\Pi$ is the number of bits communicated over worst case inputs. The deterministic communication complexity of $f \otilde g$, denoted by $\Dcc(f \otilde g)$ is defined as follows:
        \begin{align*}
            \Dcc(f \otilde g) = \min_{\Pi} \textnormal{cost}(\Pi).    
        \end{align*}
        Here, and in the following bullets, the minimization is over all protocols satisfying the correctness requirement described in the corresponding bullet.
        \item In the model of randomized communication, the players have access to an arbitrary amount of public randomness. The correctness requirement of a protocol $\Pi$ is that for all $x, y$, $(x, y, \Pi(x, y)) \notin f \circ g$ with probability at most $\eps$.  In this case we say that $\Pi$ computes $(f \otilde g)$. The cost of $\Pi$ is the number of bits communicated over worst case inputs and the private randomness. The randomized communication complexity of $f \otilde g$, denoted by $\Rcc_{\eps}(f \otilde g)$ is defined as follow:
        \begin{align*}
            \Rcc_{\eps}(f \otilde g) = \min_{\Pi} \textnormal{cost}(\Pi).
        \end{align*}
        When $\eps = 1/3$, we use the notation $\Rcc(f \otilde g)$.
        
        \item In the model of quantum communication, the correctness requirement is exactly the same as in the randomized case, but the players may use qubits to communicate. The cost of $\Pi$ is the number of qubits communicated over worst case inputs. The quantum communication complexity of $f \otilde g$, denoted by $\Qcc_{\eps}(f \otilde g)$ is defined as follow:
        \begin{align*}
            \Qcc_{\eps}(f \otilde g) = \min_{\Pi} \textnormal{cost}(\Pi).
        \end{align*}
        When $\eps = 1/3$, we use the notation $\Qcc(f \otilde g)$.
        
    \end{itemize}

Several important communication problems are of this type. Choose $f$ to be $\textsf{NOR}_n: \zone^n \to \zone$ (where for all $x \in \zone^n$, $\textsf{NOR}(x) = \overline{\OR(x)}$) and $g$ to be $\AND_2: \zone^2 \to \zone$. It is a very easy observation that the communication problem $f \otilde g$ is equivalent to the canonical Set-Disjointness problem which is defined next.

\begin{defi}[Set-Disjointness]
\label{defi: Set-Disjointness}
    Let $n > 0$ be a positive integer. The Set-Disjointness problem is denoted by $\DISJ_n : \zone^n \times \zone^n \to \zone$ and is defined by
    \begin{align*}
        \DISJ_n(A, B) = 1 \iff A \cap B = \emptyset,
    \end{align*}
    where $A, B \subseteq [n]$ are the characteristic sets of Alice and Bob's inputs, respectively.
\end{defi}

The communication complexity of $\DISJ_n$ is extensively studied. We require the following known bounds on its communication complexity~\cite{BFS86, KS92, Razborov92, razborov:qdisj, AA05}.

\begin{theorem}[Communication complexity of Set-Disjointness]
\label{thm: CC of Set-Disjointness}
    The deterministic, randomized and quantum communication complexity of $\DISJ_n$ is as follows:
    \begin{align*}
        \Dcc(\DISJ_n) = n, \quad
        \Rcc(\DISJ_n) = \Theta(n), \quad
        \Qcc(\DISJ_n) = \Theta(\sqrt{n}).
    \end{align*}
\end{theorem}

It is a folklore result that, classically, query algorithms for functions give communication protocols for these functions composed with small gadgets with very little blowup in the complexity. In the quantum setup we have the following theorem, that essentially follows from~\cite{BCW98}.

\begin{theorem}[\cite{BCW98}]
\label{thm: bcw}
Let $f \subseteq \Domain_f^n \times \cR$ be a relation where $\cD_f = [k]$ for some finite $k$, and let $g: \cD_g \times \cD_g \to \cD_f$ be a function.
For all $\epsilon > 0$, if $\Q_{\epsilon}(f) \leq T$ then $\Qcc_{\eps}(f \circ g) \leq 2T(\ceil{\log n} + \ceil{\log k} + \ceil{\log |\cD_g|})$.
\end{theorem}

We provide a proof of this theorem in Section~\ref{sec: Missing Proofs} for completeness.

A \emph{fooling set} for a communication problem $f \subseteq (\cX \times \cY) \times \cR$ is a set $S \subseteq \cX \times \cY$ such that for all pairs $s_1 = (x_1, y_1)$ and $s_2 = (x_2, y_2)$ in $S$, we have 
\begin{align*}
\cbra{r \in \cR | (x_1, y_1, r) \in f \wedge (x_1, y_2, r) \in f \wedge (x_2, y_1, r) \in f \wedge (x_2, y_2, r) \in f} = \emptyset.
\end{align*}
\begin{lemma}\label{lemma: fooling set}
    Let $f \subseteq (\cX \times \cY) \times \cR$ be a communication problem, and let $S \subseteq \cX \times \cY$ be a fooling set for $f$. Then,
    \[
    \Dcc(f) \geq \log |S|.
    \]
\end{lemma}
We refer the reader to standard texts for a formal proof~\cite{kushilevitz&nisan:cc}. We remark that standard texts usually frame the fooling set lower bound as a lower bound technique for communication complexity of functions rather than relations, but the same proof technique is easily seen to show the statement above as well. A sketch of the proof is as follows: The leaves of a protocol tree of depth $c$ yields a partition of the space $\cX \times \cY$ into $2^c$ rectangles, each of which has at least one $r \in \cR$ that is a valid output for all pairs of inputs in the rectangle. By the property of a fooling set, each element of it must belong to a different leaf. This implies the number of leaves in any protocol for $f$ must be at least $|S|$, implying that the depth of any protocol must be at least $\log |S|$.

We require the following theorem that gives an algorithm to find the maximum in a list given noisy comparison oracle access. The formulation we use below follows easily from~\cite[Theorem~15]{FPRU90}.
\begin{theorem}[{\cite[Theorem 15]{FPRU90}}]\label{thm: max noisy oracle}
    Let $S = (s_1, \dots, s_n)$ be a list of $n$ numbers. Suppose we have access to a ``noisy'' oracle, that takes as input a pair of indices $i \neq j \in [n]$, and outputs a bit that equals $\II[s_i \geq s_j]$ with probability at least $2/3$, independent of the outputs to the other queries. Then there is an algorithm that makes $O(n)$ queries to the noisy oracle and outputs $\argmax_{i \in [n]}s_i$ with probability at least $2/3$.
\end{theorem}

\begin{theorem}[\cite{Nisan93}]\label{thm: nisan GT}
    Let $n > 0$ be a positive integer. The $\GreaterThan : [n] \times [n] \to \zone$, where Alice is given $x \in [n]$ and Bob is given $y \in [n]$. is defined as $\GreaterThan(x,y) = 1$ if and only if $x \geq y$. The randomized communication complexity of $\GreaterThan$ is $O(\log\log n)$.
\end{theorem}

\subsection{Decision tree rank and decision tree size}

Let $f \subseteq \zone^n \times \cR$ be a relation.
A decision tree computing $f$ is a rooted binary tree such that: each internal node is labeled by a variable $x_i$ and has two outgoing edges, labeled $0$ and $1$, and leaf nodes are labeled by elements in $\cR$.
On input $x$, the tree’s computation starts at the root of the tree. It proceeds by computing $x_i$ as indicated by the node’s label and following the edge indicated by the value of the computed variable. The output value at the leaf, say $b \in \cR$, must be such that $(x, b) \in f$.

The deterministic query complexity of $f$ as defined earlier in this section, is easily seen to be equal to the following:
\begin{align*}
    \sD(f) := \min_{\cT:\cT~\textnormal{is a DT computing}~f} \textnormal{depth}(\cT).
\end{align*}

We next define the \emph{decision-tree size} of $f$.

\begin{defi}[Decision-tree size]\label{defn: leaf complexity}
Let $f \subseteq \zone^n \times \cR$ be a relation. Define the \emph{decision-tree size complexity} of $f$, which we denote by $\dtsize(f)$, as
\[
\dtsize(f) := \min_{\cT : \cT~\textnormal{computes}~f}\dtsize(\cT),
\]
where $\dtsize(\cT)$ denotes the number of nodes of $\cT$. \end{defi}

\begin{defi}[Decision tree rank]\label{defn: rank}
Let $\cT$ be a binary decision tree. Define the rank of $\cT$ recursively as:
For a leaf $a$, define $\rank(a) = 0$. For an internal node $u$ with children $v, w$, define
\begin{align*}
    \rank(u) = \begin{cases}
    \max\cbra{\rank(v), \rank(w)} & \textnormal{if}~\rank(v) \neq \rank(w)\\
    \rank(v) + 1 & \textnormal{if}~\rank(v) = \rank(w).
\end{cases}
\end{align*}
Define $\rank(\cT)$ to be the rank of the root of $\cT$.
\end{defi}

\begin{defi}[Rank of a Boolean function]
Let $f \subseteq \zone^n \times \cR$ be a relation. Define the rank of $f$, which we denote by $\rank(f)$, by
\[
\rank(f) = \min_{\cT : \cT~\textnormal{computes}~f} \rank(\cT).
\]
\end{defi}

We require the equivalence of rank of a Boolean function and the value of an associated Prover-Delayer game introduced by Pudlák and Impagliazzo~\cite{PI00}.
The game is played between two players: a Prover and a Delayer, who construct a partial assignment, say $\rho \in \cbra{0, 1, \bot}^n$, round-by-round. To begin with, the assignment is empty, i.e., $\rho = \bot^n$. In a round, the Prover queries an index $i \in [n]$ for which the value $x_i$ is not set in $\rho$ (i.e., $\rho_i = \bot$). The Delayer has three choices:
\begin{itemize}
    \item they either answer $x_i = 0$ or answer $x_i = 1$, or
    \item they defer the choice to the Prover.
\end{itemize}
In the latter case, the Delayer scores a point. The game ends when $f|_\rho$ is a constant function, i.e., when the Prover knows the value of the function. The value of the game, $\val(f)$, is the maximum number of points the Delayer can score over all Prover strategies.
The following result is implicit in~\cite{PI00} (also see~\cite[Theorem~3.1]{DM21} for an explicit statement and proof).
\begin{claim}\label{claim: rank equals Prover Delayer game value}
Let $f : \zone^n \times \cR$ be a relation. Then,
\begin{align*}
\rank(f) = \val(f).
\end{align*}
\end{claim}
Thus, showing a rank upper bound of $u$ amounts to giving a Prover strategy such that the Delayer cannot score more than $u$ points, and showing a lower bound of $\ell$ amounts to giving a Delayer strategy that always scores at least $\ell$ points for every Prover strategy.

A deterministic query algorithm for $f \subseteq \Domain^n \times \cR$ can equivalently be seen as a decision tree, which we define next. 

The following result due to Ehrenfeucht and Haussler relates decision tree rank to decision tree size (also see~\cite[Proposition 2.7]{DM21}). While the previous results are stated for functions, the form below is easily seen to hold when $f$ is a relation as well.
\begin{proposition}[{\cite[Lemma 1]{EH89}}]\label{prop: rank size}
    Let $f \subseteq \zone^n \times \cR$ be a relation. Then,
    \begin{align*}
        \rank(f) \leq \log \dtsize(f) \leq \rank(f) \log n.
    \end{align*}
\end{proposition}

\subsection{Formal definitions of graph problems of interest}
For clarity and completeness, we include formal definitions of the tasks of finding a king and finding a maximum out-degree vertex in this section.
\begin{defi}\label{defn: king}
    Let $n > 0$ be a positive integer. Define $\king_n \subseteq \zone^{\binom{n}{2}} \times [n]$ to be 
    \begin{align*}
        (G, v) \in \king_n \iff v \textnormal{ is a king in the tournament } G.
    \end{align*}
\end{defi}

\begin{defi}\label{defn: mod}
    Let $n > 0$ be a positive integer. Define $\MOD_n \subseteq \zone^{\binom{n}{2}} \times [n]$ to be
    \begin{align*}
        (G, v) \in \MOD_n \iff v \textnormal{ is a maximum out-degree vertex in the tournament } G.
    \end{align*}    
\end{defi}
When we give communication upper bounds for these problems, our upper bounds hold for all partitions of the input variables $\binom{n}{2}$ between Alice and Bob. When we give lower bounds, we exhibit specific partitions for which our lower bounds hold.

\section{Communication complexity of finding a source}
\label{sec: Communication complexity of finding a source}

We consider the communication complexity of finding a source in a tournament if it exists. Alice owns the edge directions a subset $E_A$ of the edges of a tournament $T \in \zone^{\binom{n}{2}}$, Bob owns the directions of the remaining edges $E_B$, and their goal is to output the label of a source in the whole tournament if it exists, or output that the tournament has no source. 
Formally, for a partition of edges $E_A, E_B$ of the complete $n$-vertex graph, define
\begin{align}
    \src_{E_A} \colon \zone^{E_A} \times \zone^{E_B} \to \{0,1,\ldots ,n\}
\end{align}
to be $\src_{E_A}(a, b) = 0$ if there is no source in the tournament defined by edge directions $a, b$, and $\src_{E_A}(a, b) = i$ if vertex $i$ is the (unique) source in the same tournament. We define the decision version of this problem to be $\srcdec_{E_A}: \zone^{E_A} \times \zone^{E_B} \to \zone$. That is, $\srcdec_{E_A}$ outputs 0 if there is no source in the tournament, and outputs 1 if there is a source.

Below, we define the celebrated Clique vs.~Independent Set problem on an $n$-vertex graph $G$~\cite{Yan91}, which we henceforth abbreviate as $\cis_G$.
The $\cis_G$ problem is associated with an $n$-vertex undirected graph $G = (V, E)$. In this problem, Alice and Bob both know $G$. Alice is given as input a clique $x \subseteq [n]$ in $G$, Bob is given as input an independent set $y \subseteq [n]$, and their goal is to either output that $x \cap y = \emptyset$, or output the label of the (unique) vertex $v$ with $\cbra{v} = x \cap y$.\footnote{Conventionally, the Clique vs.~Independent Set problem is phrased as a decision problem, where the task is to determine if $x \cap y$ is empty or non-empty. The known bounds we state here are easily seen to hold for the ``search version'' that we consider as well.}

There has been a plethora of work on the Clique vs.~Independent set problem, see for example,~\cite{Yan91, Goos15, GPW15, BBGJK21}. Of relevance to us is Theorem~\ref{thm: cis known}, which gives near-tight bounds on the deterministic communication complexity of this problem.

Perhaps surprisingly, we show that the communication problem of finding a source in a tournament is \emph{equivalent} to the Clique vs.~Independent Set problem.
Corollary~\ref{cor: CC of source} would then immediately follow.
We now prove Theorem~\ref{thm: cis equivalent to src}.

\begin{proof}[Proof of Theorem~\ref{thm: cis equivalent to src}]
In this proof, we assume for convenience that in the Clique vs.~Independent Set Problem, Alice is given an independent set and Bob is given a clique.
    \begin{itemize}
        \item Let $G = (V, E)$ be an $n$-vertex graph. Let $x, y \subseteq [n]$ be Alice and Bob's input to $\cis_G$, respectively. Recall that the vertices in $x$ form an independent set in $G$ and the vertices in $y$ form a clique in $G$. We now describe the reduction from $\cis_G$ to $\src_E$. Before going into the main reduction, we do a preprocessing of small communication cost to make sure that $G$ is connected and the size of the independent set $x$ is at least $3$. 

        Preprocessing: Bob sends the label of the connected component in $G$ that his clique $y$ is part of. Alice removes those vertices from her independent set $x$ that aren't part of this connected component. She now sends a bit to Bob to indicate whether $|x| \geq 3$. If not, she further sends labels of the two vertices in $x$ to Bob who then responds with an answer. This requires a total of $O(\log n)$ communication cost. We can therefore assume that the graph $G$ is connected and $|x|\geq 3$ for rest of the reduction.     
        Alice and Bob locally construct the following inputs to $\src_E$ (recall that Alice must construct edge directions in $E$, and Bob must construct the remaining edge directions).
        \begin{itemize}
        \item   Alice orients the edges in $E$, using Claim~\ref{claim:orienting-tree} and the fact that $G$ is a connected graph, such that only the vertices in $x$ have in-degree $0$.  
        \item Bob orients the edges in $\overline{E}$ as follows. 
        For vertices in $y$, he orients the edges in their connected components in $\overline{G}$, using Claim~\ref{claim:orienting-tree}, such that only the vertices in $y$ have in-degree $0$. Next he orients the edges of connected components that don't contain vertices of $y$. If this connected component is not a tree, he uses Claim~\ref{claim:orienting-graph-with-cycles} to orient the edges such that no vertex has in-degree $0$.  If the connected component is a tree, he orients the edges in an arbitrary way.
        \end{itemize}
        Let $T$ denote the tournament constructed above. We next show that $(x, y)$ is a $1$-input to $\cis_G$ iff there exists a source in $T$. This would prove the first part of the theorem. Moreover, we show that when there is a source in the constructed tournament, the source vertex is the same as the unique vertex in $x \cap y$.

        Let $(x,y)$ be a $1$-input to $\cis_G$ and $s$ be the unique vertex in $x\cap y$. We show that $s$ is the source in the tournament $T$. By construction, the neighbours of $s$ in $E$ are the outneighbours of $s$ in Alice's input, and the neighbours of $s$ in $\overline{E}$ are the outneighbours of $s$ in Bob's input. 

        We prove the contrapositive for the other direction. 
        Let $(x,y)$ be a $0$-input to $\cis_G$, i.e., $x\cap y = \emptyset$. We show that there is no source in $T$. Vertices in $\overline{x}$ are ruled out from being a source by the orientation of Alice's edges. Now the vertices of $x$ forms a clique in Bob's input, thus they form a connected component that is not a tree (since $|x|\geq 3$). Since this connected component does not contain a single vertex from $y$ (since we assumed $x \cap y = \emptyset$), the construction above (using Claim~\ref{claim:orienting-graph-with-cycles}) implies that all vertices in $x$ have in-degree at least 1 w.r.t.~Bob's edge directions. Thus, there is no source in the entire tournament.

        \item In the other direction, let $\zone^{E_A}$ and $\zone^{E_B}$ be Alice and Bob's input to $\src_{E_A}$, where $E_A, E_B$ form a partition of the edges of the $n$-vertex complete graph. Say that the tournament formed by these inputs is $T$. Alice and Bob construct the following instance to the Clique vs.~Independent Set problem.
        \begin{itemize}
            \item The graph is $G = (V, E)$ with $V = [n]$ and $E = E_A$.
            \item Alice constructs $x \subseteq [n]$ to be all of the vertices with in-degree 0 w.r.t.~$E_A$. It is easy to see that $x$ forms an independent set in $G$ since any edge between vertices in $x$ causes one of the vertices in $x$ to have in-degree at least 1.
            \item Bob constructs $y \subseteq [n]$ to be all of the vertices with in-degree 0 w.r.t.~$E_B$. As in the previous bullet, it is easy to see that $y$ forms an independent set in $\bar{G}$, and hence a clique in $G$.
        \end{itemize}
        Consider the input $(x, y)$ to $\cis_G$ as constructed above. We show now that $x \cap y \neq \emptyset$ iff there is a source in $T$, which would prove the second part of the theorem since $(x, y)$ and $G$ were constructed using no communication.

        Suppose $s$ is a source in $T$. Since $s$ has in-degree 0 w.r.t.~both $E_A$ and $E_B$, we must have $s \in x \cap y$. Moreover, since every other vertex must have in-degree at least 1, such a vertex is either not in $x$ or not in $y$. Thus, $s = x \cap y$. In the other direction, suppose $s = x \cap y$. By the construction above, $s$ must have in-degree 0 w.r.t.~both $E_A$ and $E_B$, and hence is a source in $T$.
    \end{itemize}
\end{proof}

\begin{claim}
\label{claim:orienting-tree}
    Let $T$ be a tree, $V$ be its vertex set and $I$ be an independent set in $T$. Then there exists an orientation of the edges of $T$ such that exactly the vertices in $V\setminus I$ have in-degree at least $1$.
\end{claim}

\begin{proof}[Proof of Claim~\ref{claim:orienting-tree}]
    We now show a procedure to orient the edges such that the set of vertices with in-degree $0$ equals the set $I$. Consider a (left-to-right) listing of subsets of vertices based on their distances from the set $I$. So if the listing looks like $V_0,V_1,\cdots ,V_j,\cdots$, then $V_0 =I$, and $V_j \subseteq V\setminus I$ is the set of vertices such that the length of a shortest path to reach a vertex in $I$ equals $j$. We orient the edges from $V_i \to V_{i+1}$ for $i\geq 0$. The edges within a partition, say $V_i$, are oriented arbitrarily. Now using the fact that tree is a connected graph, it is  easily seen that every vertex in $V\setminus I$ has in-degree at least $1$. Moreover, by our construction, all vertices in $V_0 = I$ has in-degree 0.
\end{proof}

\begin{claim}
    \label{claim:orienting-graph-with-cycles}
    Let $G$ be a connected graph that is not a tree. Then, there exists an orientation of the edges of $G$ such that every vertex of $G$ has in-degree at least $1$. 
\end{claim}
\begin{proof}[Proof of Claim~\ref{claim:orienting-graph-with-cycles}]
    Since $G$ is connected but not a tree, it contains a cycle, say $C$.  Orient the edges of $C$ in a cyclic way to give in-degree $1$ to every vertex in $C$, and then orient the edges ``away'' from the cycle $C$ (in a manner similar to the proof in Claim~\ref{claim:orienting-tree} where $V_0 = C$ here) to add $1$ to in-degrees of vertices in $V\setminus C$. Thus the directed graph so constructed has \emph{no} vertex with in-degree $0$.  
\end{proof}

\section{Communication complexity of $\king$}
\label{sec: proof of CC of Kings}

The proof of Theorem~\ref{thm: CC of Kings} is divided into two parts. We show the upper bounds in Section~\ref{sec: upper bound CC king} and the lower bounds in Section~\ref{sec: lower bound cc kings}.

\subsection{Upper bounds on communication complexity of $\king_n$}
\label{sec: upper bound CC king}
We start by proving an $O(n)$ upper bound on the deterministic communication complexity which also implies an $O(n)$ upper bound on the randomized communication complexity.
\begin{lemma}
\label{lemma: upper bound kings}
    Let $G \in  \zone^{\binom{n}{2}}$ be a tournament and let $E_1, E_2$ be a partition of the edges of $G$. The deterministic and randomized communication complexity of finding a king of $G$, where Alice is given $E_1$ and Bob is given $E_2$, is upper bounded as follows
    \begin{align*}
        \Dcc(\king_n) = O(n), \qquad \Rcc(\king_n) = O(n).
    \end{align*}
\end{lemma}
\begin{proof}
    The proof follows via the Protocol in Algorithm~\ref{algo: deterministic communication of king UB}. 

\begin{algorithm}[h]
\begin{algorithmic}[1]
\State \textbf{Input:} 
Let $G \in \zone^{\binom{n}{2}}$ be a tournament and $E_1, E_2 \subseteq \cbra{(i, j) : i < j \in [n]}$ be a partition of the edges of $G$. Alice (Player 1) is given $\zone^{E_1}$ and Bob (Player 2) is given $\zone^{E_2}$. 
\State{$S = [n]$}

\While{$|E_1| > n$ \textbf{and} $|E_2| > n$}\label{line: classical while}

\State{$b \gets \argmax_{i \in \zone}|E_i|$} \Comment{Ties broken arbitrarily}

\State{$v \gets \argmax_{w \in [n]}\cbra{\textnormal{out-degree}(w) \textnormal{ in } E_b}$} \Comment{Ties broken arbitrarily}

\State{Player $b$ sends to Player $1 - b$ the label of $v$ along with a $|S|$-bit indicator vector of the in-neighbourhood of $v$ in $E_b$} \label{line: DCC players restrict to max out-degree vertex}

\State{Player $1 - b$ sends an $|S|$-bit indicator vector of the in-neighbourhood of $v$ in $E_{1 - b}$} \label{line: DCC send in-neighouhood}

\State{$S \gets S \cap N^-(v)$} \label{line: DCC update S}
\State{$E_1 \gets $ the edges of $E_1$ that are present in $G|_S$} \label{line: DCC update E1}
\State{$E_2 \gets $ the edges of $E_2$ that are present in $G|_S$} \label{line: DCC update E2}

\EndWhile
\If{$|E_1| \leq n$}
    \State{Alice sends $E_1$ to Bob} \label{line: DCC Alice sends at most n edges to Bob}
    \State{Bob outputs a king of the tournament.}
\ElsIf{$|E_2| \leq n$}
    \State{Bob sends $E_1$ to Alice} \label{line: DCC Bob sends at most n edges to Alice}
    \State{Alice outputs a king of the tournament.}
\EndIf
\caption{Deterministic Communication Protocol for $\king_n$}
\label{algo: deterministic communication of king UB}
\end{algorithmic}
\end{algorithm}

\paragraph*{Correctness.} It is easy to see that in every iteration of the \textbf{while} loop, the size of either $E_1$ or $E_2$ decreases by at least $1$. This shows that our algorithm always terminates. 

Let $S^{(i)}$ denote the set $S$ in $i$'th iteration of the \textbf{while} loop, where $S^{(1)} = [n]$. We maintain the invariant that in every iteration of the \textbf{while} loop, a king in $G|_{S^{(i+1)}}$ is also a king in $G|_{S^{(i)}}$. This follows easily from Lemma~\ref{lem: maurer king in inneighbour} since $S^{(i+1)}$ is obtained from $S^{(i)}$ by restricting to vertices in the in-neighbourhood of some vertex $v$ in Line~\ref{line: DCC update S}. Assume without loss of generality that the \textbf{while} loop terminates with $|E_1| \leq n$. In this case, in Line~\ref{line: DCC Alice sends at most n edges to Bob}, Alice sends her edges to Bob who outputs a king of $G$.

\paragraph*{Cost.} We show that the cost of Protocol~\ref{algo: deterministic communication of king UB} is upper bounded by $O(n)$ for all tournaments $G \in \zone^{\binom{n}{2}}$. 
Suppose we enter the \textbf{while} loop with $|S| = k$.
Let $c(k)$ be the number of bits communicated during the execution of the \textbf{while} loop. Consider Line~\ref{line: DCC players restrict to max out-degree vertex}, and assume without loss of generality that $|E_1| \geq |E_2|$, thus $|E_1| \geq (1/2 \cdot \binom{k}{2})$. Since every edge in $E_1$ is an out-edge for some vertex (note that $E_1$ and $E_2$ are subsets of edges of $G|_S$ due to Line~\ref{line: DCC update E1} and Line~\ref{line: DCC update E2}) we have $\sum_{u \in S} d^+(v) \geq (1/2 \cdot \binom{k}{2})$ (where the out-degrees are only computed in $E_1$) and hence by an averaging argument there exists $v \in S$ such that the out-degree of $v$ when restricted to $E_1$ (and therefore $S$) is at least $(k-1)/4$. Thus the in-degree of $v$ in $S$ is at most $(3/4 \cdot (k-1))$.
Furthermore, in each iteration of the \textbf{while} loop, $\ceil{\log k} + k$ bits are communicated in Line~\ref{line: DCC players restrict to max out-degree vertex} and $k$ bits are communicated in Line~\ref{line: DCC send in-neighouhood}. We have the following upper bound on $c(n)$:
\begin{align*}
    c(n) \leq c(3n/4) + \ceil{\log n} + 2n,
\end{align*}
and thus $c(n) = O(n)$. Also observe that either Line~\ref{line: DCC Alice sends at most n edges to Bob} or Line~\ref{line: DCC Bob sends at most n edges to Alice} is executed and in each case at most $n$ bits are communicated. Thus the overall number of bits communicated in $O(n)$.
\end{proof}

Next, we give an $O(\sqrt{n}~\polylog(n))$ cost quantum communication protocol for $\king_n$. Our quantum communication upper bound is a corollary of Theorem~\ref{thm: bcw} (which shows how to simulate a quantum query algorithm using a quantum communication protocol) and the following theorem (which gives an $O(\sqrt{n} \cdot \polylog(n))$ quantum query algorithm for finding a king in a tournament $G \in \zone^{\binom{n}{2}}$.

\begin{theorem}[\cite{MPS23}]
\label{thm: MPS23 quantum query algo kings}
    For all $n \in \NN$, $\Q(\king_n) = O(\sqrt{n} ~\polylog(n))$.
\end{theorem}

\begin{lemma}
\label{lemma: quantum communication upper bound kings}
    Let $G \in \zone^{\binom{n}{2}}$ be a tournament and let $E_1, E_2$ be a partition of $E$. The quantum communication complexity, where Alice is given $E_1$ and Bob is given $E_2$. Then
    \begin{align*}
        \Qcc(\king_n) = O(\sqrt{n} \cdot \polylog(n)).
    \end{align*}
\end{lemma}
\begin{proof}
Given $E_1$, Alice constructs $G_1 \in \zone^{\binom{n}{2}}$ such that for all $i < j \in [n]$, and 
\begin{align*}
    (G_1)_{ij} = 
    \begin{cases}
        (E_1)_{ij} \text{ if } \{i,j\} \in E_1 \\
        0 \text{ otherwise.}
    \end{cases}
\end{align*}
Similarly Bob constructs $G_2 \in \zone^{\binom{n}{2}}$. Since $E_1, E_2$ is a partition of the edges of the tournament, observe that for all $i < j \in [n]$, $G_{ij} = \OR_2((G_1)_{ij}, (G_2)_{ij})$. 

The quantum communication protocol now follows from Theorem~\ref{thm: MPS23 quantum query algo kings} and Theorem~\ref{thm: bcw} by choosing $f = \king_n \subseteq \zone^{\binom{n}{2}} \times [n]$ as in Definition~\ref{defn: king} and $g$ to be  $\OR_2$.
\end{proof}

\subsection{Lower bounds on communication complexity of $\king_n$}\label{sec: lower bound cc kings}
Next, we prove the lower bound. In order to do this, we first give a lower bound on the communication complexity of $\PMF_n$. Recall that, in this problem, Alice is given as input a subset $S$ of $[n]$, Bob is given a ranking of elements of $[n]$ defined by $\sigma$, and their goal is to output the element in $S$ that has largest rank according to $\sigma$.

\begin{lemma}\label{lemma: pmf lower bound}
    The deterministic, randomized and quantum communication complexity of $\PMF_n$ is as follows:
    \begin{align*}
        \Dcc(\PMF_n) = \Omega(n),\quad
        \Rcc(\PMF_n) = \Omega(n),\quad
        \Qcc(\PMF_n) = \Omega(\sqrt{n}).
    \end{align*}
\end{lemma}
\begin{proof}
    We show that Set-Disjointness reduces to $\PMF_n$ and the lemma follows from Theorem~\ref{thm: CC of Set-Disjointness}. We describe the reduction next.

    Consider an input to Set-Disjointness, $S, T \subseteq [n]$ where $S$ is with Alice and $T$ is with Bob. Alice and Bob locally construct the following instance of $\PMF_n$: Alice retains her set $S$, and Bob creates an arbitrary $\sigma$ such that the following holds:
    \[
    \forall i \neq j \in [n], \qquad (T_i = 0) \wedge (T_j = 1) \implies \sigma(i) < \sigma(j).
    \]
    In other words, Bob creates a permutation $\sigma$ of $[n]$ that ranks all of the indices in $T$ higher than all of the indices outside $T$.
    They then run a protocol for $\PMF_n$ with inputs $S, \sigma$, let $k$ be the output of this protocol. If $k \in T$ then they return $S \cap T \neq \emptyset$ else they return $S \cap T = \emptyset$.

    \paragraph*{Correctness.} If $\PMF_n(S, \sigma) = \bot$, then the players know (without any additional communication) that $S = \emptyset$ and hence $\DISJ_n(S, T) = 1$. Thus, we may assume $S \neq \emptyset$. Since any protocol for $\PMF_n$ must output an index in $S$, $k \in S$.
    By Bob's construction of $\sigma$, the elements of $T$ are ranked higher than elements that are not in $T$. Since $k$ is the output of a protocol for $\PMF_n$, $k$ is the highest ranked element in $S$ by $\sigma$. Thus if $k$ is not among the top $|T|$ ranked elements, then all elements of $S$ are ranked lower than all elements of $T$ (by Bob's construction of $\sigma$) and $S \cap T = \emptyset$. On the other hand if $k$ is among the top $|T|$ ranked elements then $k \in T \cap S$. These conditions can be checked by Bob who has $\sigma$ and $k$.
\end{proof}

By the equivalence of $\PMF$ and the transitive variant of $\IndexKing$ (Observation~\ref{obs: equiv PMF IndexKing}), Lemma~\ref{lemma: pmf lower bound} implies the same lower bounds on $\tIndexKing_n$. 

We thus immediately conclude the same lower bounds on the general $\IndexKing$ problem (where Bob's tournament is arbitrary, and need not be transitive).
\begin{corollary}
\label{coro: indexking is hard}
    The deterministic, randomized and quantum communication complexity of $\IndexKing_n$ is as follows:
    \begin{align*}
        \Dcc(\IndexKing_n) = \Omega(n),\quad
        \Rcc(\IndexKing_n) = \Omega(n),\quad
        \Qcc(\IndexKing_n) = \Omega(\sqrt{n}).
    \end{align*}
\end{corollary}

We now give a lower bound on the communication complexity of $\king_n$. For this we first define a class of tournaments that we use in our proof.

\subsection{A class of tournaments}\label{sec: G S sigma}
In this section, we define a special class of tournaments on $3n$ vertices, that are parametrized by a subset $S \subseteq [n]$ and an ordering $\sigma$ of $[n]$.

\begin{defi}\label{def: G S sigma}
    Given a set $S \subseteq [n]$ and $\sigma \in \cS_n$, define the tournament $G_{S, \sigma}$ on $3n$ vertices as follows:
    \begin{itemize}
        \item The vertex set is $V = \cbra{i_b : i \in [n], b \in \cbra{0, 1, 2}}$.
        \item For each $b \in \cbra{0, 1, 2}$ and all $i \neq j \in [n]$, the direction of the edge between $i_b$ and $j_b$ is $i_b \rightarrow j_b$ iff $\sigma(i) > \sigma(j)$. We refer to these as Type 1 edges.
        \item For all $b \neq b' \in \cbra{0, 1, 2}$, all $i \in S$ and all $j \notin S$, $i_b \rightarrow j_{b'}$ is an edge. We refer to these as Type 2 edges.
        \item For all $b \neq b' \in \cbra{0, 1, 2}$ and all $i \neq j \in S$, the direction between the edge $i_b$ and $j_{b'}$ is $i_b \rightarrow j_{b'}$ iff $b' = b+1 (\textnormal{mod } 3)$. We refer to these as Type 3 edges.
        \item For all $b \neq b' \in \cbra{0, 1, 2}$ and all $i \neq j \notin S$, the direction between the edge $i_b$ and $j_{b'}$ is $i_b \rightarrow j_{b'}$ iff $b' = b+1 (\textnormal{mod } 3)$. We refer to these as Type 4 edges.
    \end{itemize}
\end{defi}
We refer the reader to Figure~\ref{fig: G S sigma} for a pictorial representation and some additional notation.

\begin{figure}
\begin{center}
    \begin{tikzpicture}[scale = .44]
        \draw [fill=gray, draw=black] (-5,0) circle (1.2cm) node (S0) {$S_0$};
        \draw [draw=black, thick] (-5,0) circle (3.5cm) node (T0) {};
        \node (T0label) [above=.7 of T0]{$T_0$};

        \draw [fill=gray, draw=black] (5,0) circle (1.2cm) node (S2) {$S_2$};
        \draw [draw=black, thick] (5,0) circle (3.5cm) node (T2) {};
        \node (T2label) [above=.7 of T2]{$T_2$};

        \draw [fill=gray, draw=black] (0,7) circle (1.2cm) node (S1) {$S_1$};
        \draw [draw=black, thick] (0,7) circle (3.5cm) node (T1) {};
        \node (T1label) [above=.7 of T1]{$T_1$};

        \draw [-{Latex[length=3mm]}, color=red, thick] (S0) -- (S1);
        \draw [-{Latex[length=3mm]}, color=red, thick] (S1) -- (S2);
        \draw [-{Latex[length=3mm]}, color=red, thick] (S2) -- (S0);

        \draw [-{Latex[length=3mm]}, color = blue, thick] (S0) -- (T1label);
        \draw [-{Latex[length=3mm]}, color = blue, thick] (S0) -- (T2label);

        \draw [-{Latex[length=3mm]}, color = olive, thick] (T0label) -- (T1label);
        \draw [-{Latex[length=3mm]}, color = olive, thick] (T1label) -- (T2label);
        \draw [-{Latex[length=3mm]}, color = olive, thick] (T2label) -- (T0label);
    \end{tikzpicture}
\end{center}
\captionsetup{singlelinecheck=off}
\caption[]{Visual depiction of $G_{S, \sigma}$. For each $b \in \cbra{0, 1, 2}$, $S_b$ contains the vertices $\cbra{i_b : i \in S}$ and $T_b$ contains the vertices $\cbra{i_b : i \notin S}$. There are four types of edges (also see Definition~\ref{def: G S sigma}): 
\begin{itemize}
    \item Edges of Type 1 are those within each $T_b \cup S_b$, here $i_b \rightarrow j_b$ iff $\sigma(i) > \sigma(j)$.
    \item Edges of \textcolor{blue}{Type 2} are those between $S_b$ and $T_{b'}$ for $b \neq b'$, here $i_b \rightarrow j_{b'}$ for all $b \neq b'$.
    \item Edges of \textcolor{red}{Type 3} are those between $S_b$ and $S_{b'}$ for $b \neq b'$, here $i_b \rightarrow j_{b'}$ iff $b' = b+1~(\textnormal{mod } 3)$. 
    \item Edges of \textcolor{olive}{Type 4} are those between $T_b$ and $T_{b'}$ for $b \neq b'$, here $i_b \rightarrow j_{b'}$ iff $b' = b+1~(\textnormal{mod } 3)$.
\end{itemize}
}
\label{fig: G S sigma}
\end{figure}

\begin{lemma}
\label{lemma: 3 kings in G S sigma}
    Let $n > 0$ be a positive integer, $S \subseteq [n]$ and $\sigma \in \cS_n$. Then, the tournament $G_{S, \sigma}$ has exactly three kings, namely $k_0, k_1, k_2$, where $k = \argmax_{j \in S}\sigma(j)$.
    Moreover, $k_0, k_1, k_2$ are the only vertices with maximum out-degree in $G_{S, \sigma}$.    
\end{lemma}
\begin{proof}
    We first show that $k_0$ is a king. The argument for $k_1, k_2$ being kings follows similarly.
    To show that $k_0$ is a king, we exhibit paths of length one or two from $k_0$ to all other vertices in the tournament.
    \begin{itemize}
        \item First note that for any element $j \in S$, there is an edge from $k_0$ to $j_0$ since $k = \argmax_{j \in S}\sigma(j)$ (this is an edge of Type 1). Thus, $k_0$ $1$-step dominates $S_0$.
        
        \item For all $j \notin S$ and $b \in \cbra{1, 2}$, there is an edge (of Type 2) from $k_0$ to $j_b$. Thus, $k_0$ $1$-step dominates $T_1$ and $T_2$.
        
        \item For $j, j' \in S$, there is an edge (of Type 3) from $k_0$ to $j_1$. Thus $k_0$ $1$-step dominates $S_1$. There is also an edge (also of Type 3) from $j_1$ to $j'_2$. Thus, $k_0$ 2-step dominates $S_2$.
        
        \item For an arbitrary $j \in S$, as noted above, there is an edge from $k_0$ to $j_1$. For $j' \notin S$, there is an edge (of Type 2) from $j_1$ to $j'_0$. Thus, $k_0$ 2-step dominates $T_0$.
    \end{itemize}
    This shows that $k_0$ (and similarly $k_1$ and $k_2$) is a king in $G_{S, \sigma}$.\footnote{We remark here that there is an alternative proof that shows $k_0$ to be a king: consider an arbitrary $j_1$ for an arbitrary $j \in S$. The in-neighborhood of $j_1$ contains $S_0$ and a subset of $S_1 \cup T_1$. It can be verified that $k_0$ is a source (and hence a king) in the tournament restricted to the in-neighbourhood of $j_1$. Lemma~\ref{lem: maurer king in inneighbour} then implies that $k_0$ is a king. We choose to keep the current proof for clarity.}
    We next show that no other vertex is a king. We do this by showing for every other vertex $k'_b$, a vertex that is not $1$-step or $2$-step dominated by $k'_b$.
    \begin{itemize}
        \item Consider $k' \neq k \in S$ and $b \in \cbra{0, 1, 2}$. We now show that $k'_b$ does not $1$-step or $2$-step dominate $k_b$.
        \begin{itemize}
            \item Since $k_b$ is the unique king in the transitive tournament $(G_{S,\sigma})|_{S_b}$ (see Lemma~\ref{lemma: Properties of Transitive Tournaments}), $k'_b$ does not $1$-step dominate $k_b$ via Type 1 edges. Moreover, the only vertices that  are $1$-step dominated by $k'_b$ via Type 1 edges are a subset of vertices in $S_b \cup T_b$. None of these vertices can $1$-step dominate $k_b$ since $(G_{S,\sigma})|_{S_b \cup T_b}$ is a transitive tournament. This shows that $k'_b$ cannot $1$-step dominate or $2$-step dominate $k_b$ by first using an edge of Type 1.

            \item The only other out-going edges from $k'_b$ are either of Type 2 or Type 3.

            \item Consider a Type 2 edge which goes from $k'_b$ to $T_{b+1~(\textnormal{mod } 3)}$ ($T_{b+2~(\textnormal{mod } 3})$ follows similarly). By construction, there is no edge from any vertex in $T_{b+1~(\textnormal{mod } 3)}$ to $k_b$ (see Figure~\ref{fig: G S sigma}).

            \item Now consider a Type 3 edge which goes from $k'_b$ to $S_{b+1~(\textnormal{mod } 3)}$. By construction, there is no edge from any vertex in $S_{b+1~(\textnormal{mod } 3)}$ to $k_b$ (see Figure~\ref{fig: G S sigma}).
        \end{itemize}
        
        \item Consider $k' \notin S$ and $b \in \cbra{0, 1, 2}$. We now show that $k'_b$ does not $1$-step or $2$-step dominate $k_{b + 2~\textnormal{(mod }3)}$.
        \begin{itemize}
            \item The only out-going edges from $k'_b$ are either of Type 1 or Type 4. On taking a Type 1 edge, $k'_b$ can only $1$-step dominate a subset of vertices of $S_b \cup T_b$. None of these vertices have an edge to $k_{b + 2~\textnormal{(mod }3)}$ (see Figure~\ref{fig: G S sigma}). Thus, $k'_b$ cannot $2$-step dominate $k_{b + 2~\textnormal{(mod }3)}$ by first taking a Type 1 edge.

            \item A Type 4 edge goes from $k'_b$ to a vertex in $T_{b+1~\textnormal{(mod }3)}$. By construction, no vertex in $T_{b+1~\textnormal{(mod }3)}$ has an edge to $k_{b + 2~\textnormal{(mod }3)}$ (see Figure~\ref{fig: G S sigma}).
        \end{itemize}
        
    \end{itemize}
    Finally, we observe that $k_0, k_1, k_2$ are the only three vertices with maximum out-degree in $G_{S, \sigma}$. Observe that the out-degrees of $k_0, k_1, k_2$ are all equal by symmetry. By Lemma~\ref{lemma: mod vertex king}, a vertex with maximum out-degree in $G_{S, \sigma}$ is a king in $G_{S, \sigma}$. This, along with the proof above that shows that $k_0, k_1, k_2$ are the only kings in $G_{S, \sigma}$, immediately implies that $k_0, k_1, k_2$ are the only three vertices with maximum out-degree in $G_{S, \sigma}$.
\end{proof}

\subsection{Proof of Theorem~\ref{thm: CC of Kings}}\label{sec: only proof of kings}
We now prove Theorem~\ref{thm: CC of Kings}. The upper bounds follow from the arguments in Section~\ref{sec: upper bound CC king}. For the lower bounds, we do a reduction from $\PMF$. The class of tournaments constructed in Section~\ref{sec: G S sigma}, and its properties, play a crucial role in the reduction.

\begin{proof}[Proof of Theorem~\ref{thm: CC of Kings}]
    The upper bounds follow from Lemma~\ref{lemma: upper bound kings} and Lemma~\ref{lemma: quantum communication upper bound kings}.

    For the lower bounds, consider an input $S \subseteq [n]$ to Alice and $\sigma \in \cS_n$ to Bob for $\PMF_n$. Alice and Bob jointly construct the tournament $G_{S, \sigma}$. Note that this construction is completely local and involves no communication; Alice can construct all edges of Types 2, 3 and 4, and Bob can construct all edges of Type 1 (see Figure~\ref{fig: G S sigma}). By Lemma~\ref{lemma: 3 kings in G S sigma}, there are exactly 3 kings in $G_{S, \sigma}$ and these are $\cbra{i_b : b \in \cbra{0, 1, 2}, i = \argmax_{j \in S}\sigma(j) = \PMF_n(S, \sigma)}$ (recall Definition~\ref{def: Permutation Maximum Finding}). Thus, running a protocol for $\king_{3n}$ on input $G_{S, \sigma}$ (where Alice has edges of Types 2, 3 and 4, and Bob has edges of Type 1) gives the solution to $\PMF_n(S, \sigma)$ at no additional cost. Lemma~\ref{lemma: pmf lower bound} implies the required lower bounds.
\end{proof}

\section{Communication complexity of $\MOD$}
\label{sec: CC of MOD}
 Recall that in the $\MOD_n$ communication problem, Alice and Bob are given inputs in $\zone^{E_1}$ and $\zone^{E_2}$, respectively, where $E_1$ and $E_2$ form a partition of the edge set $\binom{n}{2}$. Their goal is to output a vertex $v$ that has maximum out-degree in the tournament formed by the union of their edges.
We next prove Theorem~\ref{thm: CC of MOD}.
In this theorem we settle the communication complexity of finding a maximum out-degree vertex in a tournament in the deterministic, randomized and quantum models, up to logarithmic factors in the input size. In the deterministic model we are able to show a tight $\Theta(n \log n)$ bound.

We first define an intermediate communication problem, $\maxsum_{n, k}$, which we feel is independently interesting to study from the perspective of communication complexity.

\begin{defi}\label{def: max sum communication}
    Let $n, k > 0$ be positive integers. In the $\maxsum_{n, k}$ problem, Alice is given $A = (a_1, \dots, a_n) \in [k]^n$, Bob is given $B = (b_1, \dots, b_n) \in [k]^n$, and their goal is to output $\argmax_{j \in [n]}(a_j + b_j)$ (if there is a tie, they can output any of the tied indices).
\end{defi}
$\maxsum_{n, k}$ is easily seen to be the composition of two problems: the outer problem is $\textsf{ARGMAX}_{2k, n}$ (see Definition~\ref{defi: argmax search quantum algo})
and the inner function is $\SUM_k$ (which adds two integers in $[k]$, one with Alice and the other with Bob). It is also easy to see that $\MOD_n$ reduces to $\maxsum_{n, 2n}$: Alice and Bob can locally construct $(a_1, \dots, a_n)$ and $(b_1, \dots, b_n)$ to be the out-degree vectors of all the vertices restricted to edges in their inputs. Thus, a cost-$c$ protocol for $\maxsum_{n, 2n}$ also gives a protocol for $\MOD_n$.

We note here that our from upper bounds Theorem~\ref{thm: CC of MOD} actually give upper bounds for the more general $\maxsum_{n, k}$ problem; the deterministic, randomized and quantum communication upper bounds here are $O(n \log k), O(n \log \log k)$ and $O(\sqrt{n}\log k \log n)$, respectively. Next, we proceed to give a proof of Theorem~\ref{thm: CC of MOD}.

\begin{proof}[Proof of Theorem~\ref{thm: CC of MOD}]
    For the upper bounds, we exhibit protocols of the required cost for $\maxsum_{n, n}$, which is only a (potentially) harder problem. 
    \begin{itemize}
        \item For the deterministic upper bound, note that Alice can just send her input to Bob with cost $n \log n$, and Bob can output the answer.
        \item The randomized upper bound follows by using Theorem~\ref{thm: max noisy oracle} with the list $s = (a_1 + b_1, \dots, a_n + b_n)$, and observing that testing whether $a_i + b_i \geq a_j + b_j$ can be done with communication $O(\log \log n)$ and success probability at least $2/3$ (Theorem~\ref{thm: nisan GT}).
        \item  For the quantum upper bound, recall that $\maxsum_{n, n}$ is the composition of $\textsf{ARGMAX}_{2n, n}$ (with an input list in $[2n]^n$) and $\SUM$ (sum of 2 integers in $[n]$, one with Alice and the other with Bob). Here, $\mathsf{ARGMAX}_{2n, n}$ has query complexity $O(\sqrt{n})$, where query access is to the values of the elements of the list (see Theorem~\ref{thm: argmax grover}) and $\SUM : [n] \times [n] \to [2n]$. Setting $\cD_g = [n]$, $\cD_f = [2n]$, $g = \SUM_n : \cD_g \times \cD_g \to \cD_f$, , $f = \mathsf{ARGMAX}_{2n, n} \subseteq \cD_f^n \times [n]$ Theorem~\ref{thm: bcw} this gives a quantum communication upper bound of $O(\sqrt{n} \log n)$.
    \end{itemize}
    \paragraph*{Randomized and quantum lower bounds.}
    The randomized and quantum lower bounds follow the same proof as that of Theorem~\ref{thm: CC of Kings} (see Section~\ref{sec: only proof of kings}) because the three kings in $G_{S, \sigma}$ are precisely the maximum out-degree vertices there as well (see Lemma~\ref{lemma: 3 kings in G S sigma}). This argument also shows a deterministic lower bound of $\Omega(n)$.
    
    \paragraph*{Deterministic lower bound.}
    We now turn our attention to the deterministic lower bound of $\Omega(n \log n)$, which does not use the same reduction as in the proof of Theorem~\ref{thm: CC of Kings}.
    We show this via a fooling set argument (Lemma~\ref{lemma: fooling set}). Below, we assume that the first half of Alice's input corresponds to the out-degree sequence of a tournament on vertex set $L = \cbra{1, 2, \dots, n/2}$, the second half of her input corresponds to the out-degree sequence of a tournament on vertex set $R = \cbra{1', 2', \dots, (n/2)'}$, and Bob's input is the out-degree sequence of the complete bipartite graph between $L$ and $R$. We focus on inputs that are induced by tournaments of the following form, that are defined for a permutation $\sigma \in \cS_{n/2 - 1}$ that acts in an identical fashion on $\cbra{2, 3, \dots, n/2}$ and $\cbra{2', 3', \dots (n/2)'}$. We call Alice and Bob's input constructed below $A_\sigma$ and $B_\sigma$, respectively.
    \begin{itemize}
        \item Vertex 1 is the source in $L$, and vertex $1'$ is the source in $R$. These edges are with Alice.\footnote{When we say ``edges are with Alice/Bob'', we actually mean Alice/Bob's out-degree of vertices is determined by the directions of the underlying edges. In this case we mean Alice's first coordinate is $n/2 + 1$ because vertex 1 is a source in $L$.}
        \item Vertex 1 has edges towards $1'$ and $\sigma^{-1}(2')$. All other vertices in $\cbra{3', 4', \dots (n/2)'}$ have edges pointing towards vertex $1$. These edges are with Bob.
        \item For all $i, j \in \cbra{2, 3, \dots, n/2}$, there is an edge from $i$ to $j$ iff $\sigma(i) < \sigma(j)$. Similarly there is an edge from $i'$ to $j'$ iff $\sigma(i') < \sigma(j')$. These edges are with Alice.
        \item For $i \in \cbra{2, 3, \dots, n/2}$, there is an edge from $i$ to $1'$. These edges are with Bob.
        \item For $i, j \in \cbra{2, 3, \dots, n/2}$, there is an edge from $i$ to $j'$ iff $\sigma(i) \leq \sigma(j)$. These edges are with Bob.
    \end{itemize}
    We now verify that vertex 1 is the unique vertex with maximum out-degree in the whole tournament (and hence the first coordinate must be output in the corresponding inputs to Alice and Bob for $\MOD_n$).
    \begin{itemize}
        \item The first two bullets above ensure that vertex 1 has out-degree $n/2 - 1 + 2 = n/2 + 1$. 
        \item The first and fourth bullets ensure that the out-degree of vertex $1'$ is $n/2 - 1$.
        \item The second and fifth bullets ensure that vertex $\sigma^{-1}(2')$ has out-degree $n/2 - 2$.
        \item For $i \in \cbra{2, 3, \dots, n/2}$, the out-degree of vertex $\sigma^{-1}(i)$ is $n/2 - i$ from Alice's input (third bullet) plus $i$ from Bob's input (fifth bullet), which gives a total of $n/2$.
        \item For $i \in \cbra{3, 4, \dots, n/2}$, the out-degree of vertex $\sigma^{-1}(i')$ is $n/2 - i$ from Alice's input (third bullet) plus $i - 1$ from Bob's input (fifth line), which gives a total of $n/2 - 1$.
    \end{itemize}
    These bullets verify that for input $(A_\sigma, B_\sigma)$, vertex 1 is the unique maximum out-degree vertex. Our fooling set will be of the form $F = \cbra{(A_\sigma, B_\sigma) : \sigma \in S}$, where $S \subseteq \cS_{n/2-1}$ is chosen appropriately. The property that $S$ will satisfy is that for all $\sigma \neq \sigma' \in S$, at least one of the inputs $(A_{\sigma}, B_{\sigma'})$ or $(A_{\sigma'}, B_\sigma)$ will \emph{not} have vertex 1 as a maximum out-degree vertex. We will also construct $S$ such that $|S| = 2^{\Omega(n \log n)}$. Lemma~\ref{lemma: fooling set} will then imply the required deterministic communication lower bound of $\Omega(n \log n)$. 
    
    It remains to construct $S \subseteq \cS_{n/2 - 1}$, which we do in the remaining part of this proof.
    We construct $S$ such that it satisfies the following property.
    \begin{align*}
        \forall \sigma \neq \sigma' \in S, \qquad \exists i \in \cbra{2, 3, \dots, n/2} : |\sigma(i) - \sigma'(i)| \geq 2.
    \end{align*}
    In the two bullets below, we first show why such an $S$ satisfies the required fooling set property, and then show a construction of $S$ of size $2^{\Omega(n \log n)}$.
    \begin{itemize}
        \item Let $\sigma \neq \sigma'$ be an arbitrary pair of elements of $S$. Without loss of generality, assume that $i \in \cbra{2, 3, \dots, n/2}$ is such that $\sigma'(i) - \sigma(i) \geq 2$ (otherwise switch the roles of $\sigma$ and $\sigma'$ and run the same argument). Consider the input $(A_\sigma, B_{\sigma'})$. Note that  the out-degree of vertex 1 remains $n/2 + 1$ because all edges incident on it are fixed for all inputs in our fooling set.
        Alice's contribution to the out-degree of vertex $i$ is $n/2 - \sigma(i)$, and Bob's contribution is $\sigma'(i)$, which gives a total of $n/2 + \sigma'(i) - \sigma(i) \geq n/2 + 2$. Thus vertex 1 cannot be a maximum out-degree vertex in the input $(A_\sigma, B_{\sigma'})$.
        \item We construct such an $S$ greedily one element at a time. At any step in the construction we maintain the invariant that the current set $T$ satisfies
        \begin{align*}
            \forall \sigma \neq \sigma' \in T, \qquad \exists i \in \cbra{2, 3, \dots, n/2} : |\sigma(i) - \sigma'(i)| \geq 2.
        \end{align*}
        Additionally we maintain a ``candidate'' set of permutations in $\cS_{n/2 - 1}$ that are not in $T$, and have the property that adding any of them to $T$ will satisfy $T$'s invariant.
        Initially we start with $T = \emptyset$ and the candidate set as $\cS_{n/2 - 1}$, which clearly satisfies the required invariant. At any stage, after adding $\sigma$ to $T$, we remove the set $S_\sigma$ from the candidate set, where $S_\sigma$ is defined as 
        \begin{align*}
            S_\sigma := \cbra{\tau \in \cS_{n/2 - 1} : |\tau(i) - \sigma(i)| < 2}~\forall i \in \cbra{2, 3, \dots, n/2}.
        \end{align*}
        It is easy to verify by induction that $T$ and the candidate set thus constructed always satisfy the required invariant. The initial size of the candidate set is $(n/2 - 1)! = 2^{\Omega(n \log n)}$, and at each step we are removing at most $3^n$ elements from the candidate set. This means that the number of iterations of this construction is at least $2^{\Omega(n \log n - n)} = 2^{\Omega(n \log n)}$, which is what we needed.
    \end{itemize}
\end{proof}
We remark that while it may seem like the argument used in the previous proof may be adaptable to prove a deterministic communication lower bound of $\Omega(n \log n)$ for $\king_n$, this is not possible in view of our $O(n)$ deterministic communication upper bound for $\king_n$ from Theorem~\ref{thm: CC of Kings}. This shows an inherent difference between $\MOD_n$ and $\king_n$ in the setting of deterministic communication complexity.

\section{Decision tree rank of $\king_n$}
In this section we prove a tight bound of $n-1$ on the decision tree rank of $\king_n$.
Recalling Claim~\ref{claim: rank equals Prover Delayer game value} and the
discussion following the claim, we show our rank upper bound by giving a Prover strategy and our lower bound by giving a Delayer strategy.
\begin{proof}[Proof of Theorem~\ref{thm: rank of king}]
We use Claim~\ref{claim: rank equals Prover Delayer game value}. We first prove the upper bound and then give a proof of the lower bound.

\paragraph*{Upper bound.}
The Prover strategy for the upper bound is given in Algorithm~\ref{algo: prover strategy}.
\begin{algorithm}[h]
\begin{algorithmic}[1]
\State{$\rho \gets \bot^{\binom{n}{2}}$}\Comment{This is the list of edge orientations known so far. Initially this is empty.}
\State{$V \gets [n]$}

\While{$V \neq \emptyset$}

\State{$v \gets$ an arbitrary vertex in $V$}

\ForAll{$u\in V\setminus \{v\}$} 

\State{Prover queries  the orientation of the undirected edge $e = (v,u)$} 

\If{Prover is given the choice}
\State{Prover directs the edge $e$ out of $v$, i.e., $v\to u$}\label{line:prover-choice} 

\Comment{Delayer scores $1$ point.}
\Else \State{Delayer chooses $e$'s direction}\label{line:delayer-choice} \Comment{Delayer scores $0$ point.}
\EndIf
\State{Update $\rho$}\Comment{Update the edge $e$ orientation, given by either Prover or Delayer.}
\EndFor

\State{$V \gets N^{-}(v) \cap V$}\label{line: go in in-nbd} \Comment{Move to in-neighbourhood of $v$.}

\EndWhile
\caption{Prover strategy}
\label{algo: prover strategy}
\end{algorithmic}
\end{algorithm}
We now analyze this strategy. 

\emph{Proof of Correctness}: The game terminates when $V = \emptyset$, which implies $N^{-1}(v) = \emptyset$ (Line~\ref{line: go in in-nbd}), which further implies $v$ is the source among the vertices (remaining) in $V$ during the last execution of the \textbf{while} loop. Now a recursive application of Lemma~\ref{lem: maurer king in inneighbour} implies that $v$ is indeed a king in the whole tournament. 

\paragraph*{Upper bound on Delayer's score.} First note that for every score that Delayer earns, Prover adds one vertex to the out-neighbour $N^{+}(v)$ of $v$. Therefore, in each execution of the \textbf{while} loop if $k$ is the score that Delayer earns then at least $k$ is the number of vertices added to $N^{+}(v)$ (it could be the case that the Delayer's choice also adds to the out-neighbours of $v$) and at least $k+1$ vertices are removed from $V$ for the next iteration. 

Let $r$ be the number of executions of $\textbf{while}$ loop before it terminates. Note that $r\geq 1$ since $V\neq \emptyset$ in the beginning. For $i\in[r]$, let $k_i$ be the score that Delayer earns in the $i$-th execution. Further let the size of the out-neighbourhood $|N^{+}(v)|$ in the $i$-th execution be $k_i + \lambda_i$ for some $\lambda_i \geq 0$. Then we have 
\(\sum_{i=1}^r (k_i+ \lambda_i + 1) = n\), which implies \(\sum_{i=1}^rk_i = n -r -(\sum_{i=1}^r\lambda_i)\). Therefore, the Delayer's score, $\sum_{i=1}^rk_i$, is at most $n-1$, since $r \geq 1$. 

\paragraph*{Lower bound.}
The Delayer's strategy to show a $n-1$ lower bound on the rank is simple: the Delayer gives the Prover the choice for the first $n-1$ queries of the Prover.
It remains to show that the Prover cannot output a king after the first $n-2$ queries. Towards a contradiction, suppose the Prover outputs a vertex $v$ to be a king after at most $n-2$ queries. Since at most $n-2$ queries has been made, there exists a partition of the vertex set into two parts, sat $L$ and $R$, such that no edges crossing the cut has been queried yet. Without loss of generality, assume $v \in L$. Now it is easy to see that there exists a tournament $G'$ consistent with the queries made so far such that all edges in the cut are directed from $R$ to $L$. Clearly $v$ is not a king in $G'$ and hence the Prover's output was incorrect. Thus, the Delayer can always score at least $n-1$ irrespective of the Prover's strategy. 
\end{proof}

\bibliography{reference}

\appendix

\section{Appendix}
\label{sec: Missing Proofs}
\subsection{Equivalence of $\PMF$ and $\tIndexKing$.} We prove Observation~\ref{obs: equiv PMF IndexKing}. We require the fact that a transitive tournament induces a total ordering of the vertices (see Lemma~\ref{lemma: Properties of Transitive Tournaments}).
\begin{proof}[Proof of Observation~\ref{obs: equiv PMF IndexKing}]
    We show that without any communication, $\PMF_n$ can be reduced to $\tIndexKing_n$ and vice versa.
    
    Let $S \subseteq [n]$ and $G \in \zone^{\binom{n}{2}}$, where $G$ is a transitive tournament, be the inputs to Alice and Bob respectively for $\tIndexKing_n$. For reduction to $\PMF_n$, Alice retains her set $S$ and Bob constructs a ranking $\sigma$ of $[n]$ such that:
    \begin{itemize}
        \item If $v$ is the $i$'th vertex in the total ordering induced by the transitive tournament, then $\sigma(v) = i$.
    \end{itemize}
Observe that $v$ is a king in $G|_S$ if and only if $v$ is the unique source vertex in $G|_S$ if and only if $v = \argmax_{w \in S} \sigma(w)$.

    Next, we show that $\PMF_n$ reduces to $\tIndexKing_n$ without any communication. Let $S \subseteq [n]$ and $\sigma \in \cS_n$ be inputs to Alice and Bob respectively. Again, Alice retains her input while Bob constructs a transitive tournament $G$ with the following properties:
    \begin{itemize}
        \item If $\sigma(v) = i$, then Bob's transitive tournament is such that $v$ is the $i$'th vertex in the induced total ordering of the vertices.
    \end{itemize}
Observe that $v = \argmax_{w \in S} \sigma(w)$ if and only if $v$ is the source in the tournament $G|_S$. \end{proof}

\subsection{From quantum query algorithms to communication protocols.}
Next, we provide a proof of Theorem~\ref{thm: bcw}, due to~\cite{BCW98}, for completeness. The proof follows the exposition of~\cite{DEWOLF2002337}.

\begin{proof}[Proof of Theorem~\ref{thm: bcw}]
Let $(x^{(1)}, \dots, x^{(n)})$ and $(y^{(1)}, \dots, y^{(n)})$ be inputs to Alice and Bob respectively where $x^{(i)}, y^{(i)} \in \cD_g$. Let
$m = \ceil{\log |\cD_g|}$.

Let $\cA$ be an $\eps$-error quantum query algorithm for $f$ of query cost $T$. To obtain a communication protocol, Alice simulates $\cA$ on input
\begin{align*}
    \left(g(x^{(1)},y^{(1)}), \dots, g(x^{(n)},y^{(n)})\right).
\end{align*}
Suppose at some point during simulation of $\cA$ Alice wants to apply query to the state $\ket{\psi} = \sum_{i, b} \alpha_{i,b} \ket{i} \ket{b}$, where $i \in [n]$ and $b \in [k]$. Here the first register has $\ceil{\log n}$ qubits and the second register has $\ceil{\log k}$ qubits. This is achieved by the following steps:
\begin{itemize}
    \item Alice attaches $\ket{0^m}$ to $\ket{\psi}$ and prepares the state $\sum_{i, b} \alpha_{i,b} \ket{i} \ket{b} \ket{x^{(i)}}$ using the unitary $\ket{i} \ket{b} \ket{0^{m}} \rightarrow \ket{i} \ket{b} \ket{x^{(i)}}$. She sends this state to Bob.

    \item Bob applies the unitary $\ket{i} \ket{b} \ket{x^{(i)}} \rightarrow \ket{i} \ket{b + g(x^{(i)}, y^{(i)}) \textnormal{ mod } k} \ket{x^{(i)}}$ and sends Alice the state 
    \begin{align*}
    \sum_{i,b} \alpha_{i,b} \ket{i} \ket{b + g(x^{(i)}, y^{(i)}) \textnormal{ mod } k} \ket{x^{(i)}}.    
    \end{align*}

    \item Alice applies the unitary $\ket{i, b, x^{(i)}} \rightarrow \ket{i, b, 0^{m}}$ an obtains the state 
    \begin{align*}
    \left(\sum_{i,b} \alpha_{i,b} \ket{i} \ket{b + g(x^{(i)}, y^{(i)}) \textnormal{ mod } k} \right) \ket{0^m}     
    \end{align*}
\end{itemize}
Thus, by communicating $2(\ceil{\log n} + \ceil{\log k} + m)$ qubits, the players have implemented quantum query on an arbitrary state exactly. This implies that $\Qcc_{\eps}(f \circ g) \leq 2T(\ceil{\log n} + \ceil{\log k} + \ceil{\log |\cD_g|})$.
\end{proof}

\subsection{Tight randomized query complexity of finding a king}
We consider the randomized query complexity of finding a king (see Section~\ref{sec: Prelims Query and Communication Complexity}).The best lower bound and upper bound for this problem (due to~\cite{MPS23}) is $\Omega(n)$ and $O(n \log\log n)$ respectively. We close this gap by giving an $O(n)$ randomized query algorithm for finding a king.

We need the following lemma.

\begin{lemma}[{\cite[Lemma 14]{MPS23}}]
\label{lemma: Out-degree of a random vertex}
    Let $G \in \zone^{\binom{n}{2}}$ be a tournament and $v \in [n]$ be chosen uniformly at random. Then $d^{-}(v) \leq 4(n-1)/5$ with probability at least $3/5$.
\end{lemma}
\begin{lemma}
\label{lemma: RQC of finding king}
Let $G \in \zone^{\binom{n}{2}}$ be a tournament. Given query access to $G$, there is a randomized query algorithm that returns a king in $G$ by making $O(n)$ queries. 
\end{lemma}

\begin{proof}
We give an algorithm (Algorithm~\ref{algo: randomized query algorithm for king}) that is correct on all inputs and has an expected cost of $O(n)$. By a standard application of Markov's inequality this gives an algorithm with worst cast cost $O(n)$ and error probability at most $1/3$ for every tournament $G \in \zone^{\binom{n}{2}}$. 

Consider Algorithm~\ref{algo: randomized query algorithm for king}. 
From Lemma~\ref{lem: maurer king in inneighbour}, it is easy to verify that the algorithm always returns a correct answer. The algorithm makes at most $n$ queries in Line~\ref{line: brute force king}. Next we upper bound the expected number of queries in the \textbf{while} loop.

Let $A(n)$ be the expected number of queries made by the \textbf{while} loop. 
Note that in Line~\ref{line: quey while}, at most $(n-1)$ queries are made.
From Lemma~\ref{lemma: Out-degree of a random vertex}, with probability at least $3/5$, $d^-(v) \leq 4n/5$ for $v$ sampled in Line~\ref{line: sample rand query algo}. Thus
\begin{align*}
    A(n) \leq (n-1) + 3/5 \cdot A(4n/5) + 2/5 \cdot A(n),
\end{align*}
which implies
\begin{align*}
    A(n) \leq 5n/3 + A(4n/5).
\end{align*}
This implies that $A(n) = O(n)$.

\begin{algorithm}[ht]
\begin{algorithmic}[1]
\State \textbf{Input:} 
Query access to a tournament $G \in \zone^{\binom{n}{2}}$  
\State{$T \gets [n]$}

\While{$|T| > \sqrt{n}$}\label{line: query algo while}

\State{$v \gets $ random vertices drawn independently from $T$}\label{line: sample rand query algo}

\State{$T \gets T \setminus (N^-(v) \cap T)$} \label{line: quey while} \Comment{Query the out-neighbours of $v$ in $T$.}
\EndWhile

\State{Return a king in $T$}\label{line: brute force king}

\caption{Randomized Query Algorithm for $\king_n$}
\label{algo: randomized query algorithm for king}
\end{algorithmic}
\end{algorithm}

\end{proof}

\end{document}